\documentclass{amsart}
\usepackage{amssymb}
\usepackage{verbatim}
\usepackage{mathrsfs}
\usepackage{dsfont}

\newtheorem{theorem}{Theorem}[section]

\newtheorem{corollary}[theorem]{Corollary}
\newtheorem{proposition}[theorem]{Proposition}

\theoremstyle{definition}
\newtheorem{definition}[theorem]{Definition}

\theoremstyle{remark}
\newtheorem{remark}[theorem]{Remark}
\newtheorem{observation}[theorem]{Observation}

\newcommand{\I}{{\mathds {1}}}
\newcommand{\cA}{{\mathcal A}}

\newcommand{\cD}{{\mathcal D}}

\newcommand{\cH}{{\mathcal H}}

\newcommand{\cK}{{\mathcal K}}
\newcommand{\cM}{{\mathcal M}}
\newcommand{\qM}{{\mathfrak M}}
\newcommand{\qN}{{\mathfrak N}}
\newcommand{\cN}{{\mathcal N}}
\newcommand{\cP}{{\mathcal P}}

\newcommand{\cO}{{\mathcal O}}
\newcommand{\cR}{{\mathcal R}}

\newcommand{\bC}{{\mathbb{C}}}

\newcommand{\Rn}{{\rm I\!R}} 
\newcommand{\Nn}{{\rm I\!N}} 
\newcommand{\Cn}{{\setbox0=\hbox{
$\displaystyle\rm C$}\hbox{\hbox
to0pt{\kern0.6\wd0\vrule height0.9\ht0\hss}\box0}}} 

\numberwithin{equation}{section}

\newcommand{\hdelta}{\tilde{\Delta}_{\omega}}

\newcommand{\jed}{{\mathbb{I}}}
\begin{document}

\title{Integral and Differential structures for quantum field theory}

\author{L. E. Labuschagne}
\address{DSI-NRF CoE in Math. and Stat. Sci,\\ Unit for BMI,\\ Internal Box 209, School of Math \& Stat. Sci.\\
NWU, PVT. BAG X6001, 2520 Potchefstroom\\ South Africa}
\email{Louis.Labuschagne@nwu.ac.za}

\author{W. A. Majewski}
\address{DSI-NRF CoE in Math. and Stat. Sci,\\ Unit for BMI,\\ Internal Box 209, School of Math. \& Stat. Sci.\\
NWU, PVT. BAG X6001, 2520 Potchefstroom\\ South Africa\newline 
and Institute of Theoretical Physics and Astrophysics, Gda{\'n}sk University, 
Wita Stwo\-sza~57, 80-952 Gda{\'n}sk, Poland} \email{fizwam@univ.gda.pl}

\date{\today}


\thanks{The contribution of L. E. Labuschagne is based on research partially supported by the National Research Foundation (IPRR Grant 96128). Any opinion, findings and conclusions or recommendations expressed in this material, are those of the author, and therefore the NRF do not accept any liability in regard thereto.}

\begin{abstract}
The aim of this work is to firstly demonstrate the efficacy of the recently proposed Orlicz space formalism for Quantum theory \cite{ML}, and secondly to show how noncommutative differential structures may naturally be incorporated into this framework. To start off with we specifically propose regularity conditions which in the context of local algebras corresponding to Minkowski space, ensure good behaviour of field operators as observables, and then show that fields obtained by the Osterwalder-Schrader reconstruction theorem are regular in this sense. This complements earlier work by Buchholz, Driessler, Summers and Wichman, etc, on generalized $H$-bounds. The pair of Orlicz spaces we explicitly use for this purpose, are respectively built on the exponential function (for the description of regular field operators) and on an entropic type function (for the description of the corresponding states). This formalism has been shown to be well suited to a description of quantum statistical mechanics, and in the present work we show that it is also a very useful and elegant tool for Quantum Field Theory. We then introduce the class of tangentially conditioned algebras, which is a large class of local algebras corresponding to globally hyperbolic Lorentzian manifolds that locally ``look like'' the local algebras of Minkowski space. On the one hand this ensures that at a local level, the Orlicz space formalism discussed above is also relevant for a much more general class of local algebras. On the other hand, the structure of this class of algebras, allows for the development of a non-commutative differential geometric structure along the lines of the du Bois-Violette approach to such a theory. In this way we obtain a complete depiction: integrable structures based on local algebras provide a static setting for an analysis of Quantum Field Theory and an effective tool for describing regular behaviour of field operators, whereas differentiable structures posit indispensable tools for a description of equations of motion.
\end{abstract}

\maketitle

\noindent
\section{Preliminaries; some basic ideas derived from QFT}
This paper may roughly be divided into two major parts. In sections 1 and 2 we explore how noncommutative integration theory may be used to harmonise competing formalisms for quantum field theory. Section 3 represent an interlude where we investigate tangential phenomena for local algebras on Lorentzian manifolds. This section also serves as a ``bridge'' joining the two major parts of the paper, the second of which is devoted to the differential structure of local algebras, and which is contained in section 4. 

These two parts are strongly related. We recall (see the discussions on the ``germs'' of a theory in Haag's book \cite[Section VIII.1, page 326]{haag}) that in local quantum field theory the structure describing the relation between the family of open subsets of space-time $\mathbb{M}$ and observables of the theories is commonly referred to as a presheaf. In particular, the notion of presheaf is naturally related to the net structure of algebras. This serves as the first pillar of the first part. The second pillar of that part is provided by noncommutative integration theory. Turning to the second part of the paper we follow the idea that Quantum Field Theory is strictly local in the sense that basic information is associated with small neighbourhoods of spacetime points. In particular the notion of germs referred to above stems from the strict locality of quantum fields, and as such is a concept that is at the heart of the notion of tangent spaces. So it seems natural to complement the development of integral structures by the development of differential structures related to more general manifolds which only locally ``look like'' $\mathbb{M}$''. Such differential structures of course form indispensable tools for the description of time evolution of quantum systems in these very general contexts.

With each of these objectives a measure of revision is required in order to set the results achieved in their proper context. Generally, it would seem that in Quantum Mechanics there are two schemes for a description of physical systems, cf. \cite{Bor1}. The first method uses bounded operators. The idea of introducing the norm topology on the set of observables was strongly advocated by I. Segal \cite{Segal}. To argue in favor of this idea one can say that in a laboratory a physicist deals with bounded functions of observables only! However, as it was already remarked by Borchers \cite{Bor1}, in this method \textit{``some detailed information about a physical system is usually lost''}. Furthermore, this scheme admits ``non-physical states'' having badly defined entropy, see \cite{Maj1} and the references given there.

The second method uses unbounded operators. The motivation for this method can be taken from representations of canonical commutation relations, Wightman's formulation of quantum field theory and the theory of Lie algebras. Although mathematical aspects of algebras of unbounded operators have been analyzed in much detail, see \cite{AIT}, \cite{Schmud}, \cite{Bag}, it is well known that formal calculations can be misleading; see Section VIII.5 in \cite{RSI}. 

In section 2, we will argue that through the addition of fairly mild regularity restrictions, the field operators may in a natural way be realised as part of the (non-commutative) integration theory of the local algebras. Besides other technical conditions, this scheme relies on the selection of ``more'' regular unbounded operators, where ``more'' regular means conditions which ensure the so-called $\tau$-measurability of the field operators (see the following pages for definitions and details). Our objective in this paper, is therefore not to analyse some specific model, but rather to on the one hand introduce this scheme, and on the other to show how various aspects of Quantum Field Theory may be harmonised within this framework. 

The algebraic approach to relativistic quantum field theory was formulated in the sixties by R. Haag, D. Kastler, H. Araki, H. J. Borchers and others, see \cite{haag}, \cite{araki}, \cite{BH}.

The basic object of this approach is a net of von Neumann algebras, $\cO \mapsto \qM(\cO)$, on a Hilbert space $\cH$, labeled by subsets $\cO$ of (Minkowski) space-time $\Rn^4$. It satisfies, see \cite{araki}:
\begin{enumerate}
\item[{(L1)}]  Isotony: $\cO_1 \subset \cO_2$ implies $\qM(\cO_1) \subset \qM(\cO_2).$
\item[{(L2)}]  Covariance: for $g= (a,\Lambda) \in P^{\uparrow}_+$, there is a representation $\alpha_g$ in $Aut(\qM)$ such that
$$\alpha_g(\qM(\cO)) = \qM(g\cO), \quad g\cO = \{ \Lambda x +a; x \in \cO\}, $$
 $P^{\uparrow}_+$ stands for the Poincar\'e group, where the Lorentz group is restricted, homogeneous, see Section 3.3 in \cite{araki}. 
This action is realised by a strong operator continuous unitary group $\{U(a, \Lambda)\}\subset \qM=\overline{\bigcup_{\cO\subset \Rn^4}\qM(\cO)}^{w^*}$. (See \cite{Bor3} for this last restriction.)
\item[{(L3)}]  Locality: if $\cO_1$ and $\cO_2$ are spacelike separated then $\qM(\cO_1)$ and $\qM(\cO_2)$ commute.
\item[{(L4)}]  Weak Additivity: 
$$\qM = \left( \bigcup_{x \in \Rn^4} \qM(\cO + x) \right)^{''}$$
for all open $\cO$.
\item[{(L5)}]  Vacuum vector: there is a normalized vector $\Omega \in \cH$, unique up to a phase factor that satisfies
$(\Omega, \alpha_g(f) \Omega) = (\Omega, f \Omega)$ for each $f$ in the global algebra $\qM=\overline{\bigcup_{\cO\subset \Rn^4}\qM(\cO)}^{w^*}$ and each $g= (a,\Lambda) \in P^{\uparrow}_+$.
\item[{(L6)}]  Positivity: The generator of translation has spectrum lying in the forward light cone.
\end{enumerate}

\begin{remark} 
We pause to comment on axiom (L5). In general the Reeh-Schlieder theorem provides criteria ensuring that the vacuum vector is both cyclic and seperating for each local algebra $\qM(\cO)$ corresponding to some bounded open region in $\mathbb{M}$. However this fact does not ensure that it is also separating for the full algebra $\qM$. The most we can say is that it is (trivially) also cyclic for this algebra. 
\end{remark}

Later on in the paper we will pass to the standard form of the algebra $\qM$ as realised on the noncommutative space $L^1(\qM)$. That framework is technically more complex. The above remark therefore serves to clarify the proof strategy of Theorem \ref{Lcosh-criteria} where we shall be working in this technically more complex context.

On the other hand, in the fifties, Wightman and G\"arding, see \cite{GW1}, \cite{GW2}, \cite{W}, 
formulated postulates for \textit{general quantum field theory} in terms of (unbounded) operators on a Hilbert space. 
Depending on the context, there are various subtle variations of these postulates. However we are not interested in a detailed application of the resultant theory to a specific context, but rather in the overarching mathematical framework and how this framework may be harmonised. The basic mathematical ingredients of these postulates that are relevant to our study, may be expressed as below (see \cite{araki}). The reader interested in finer detail, may refer to \cite{GW1}, \cite{GW2} and \cite{W}.
\begin{enumerate}
\item[{(F1)}]\label{F1}  Quantum fields: The operators $\phi_1(f), ...,\phi_n(f)$ are given for each $C^{\infty}$-function with compact support in the Minkowski space $\Rn^4$. Each $\phi_i(f)$ and its hermitian conjugate operator $\phi_j^*(f)$ are defined on at least a common dense linear subset $\cD$ of the Hilbert space $\cH$ and $\cD$ satisfies
$$\phi_j(f)\cD\subset \cD, \quad \phi_j^*(f) \cD \subset \cD,$$
for any $f$, $j=1,...,n$. For any $v,w \in \cD$
$$f \mapsto (v, \phi_j(f)w)$$
is a complex valued distribution.
\item[{(F2)}] Relativistic symmetry: There is a well defined strongly continuous unitary representation $U(a, A)$ of $P^{\uparrow}_+$ ($a \in \Rn^4, \ A \in SL(2, \Cn)$) such that
$$U(a, \Lambda) \cD = \cD$$
and
$$U(a, A)\phi_j(f) U(a, A)^* = \sum S(A^{-1})_{jk} \phi_k(f_{(a,A)}),$$
where the matrix $(S(A)_{j,k})$ is $n$-dimensional representation of $A \in SL(2,\Cn)$, and $f_{(a,A)}(z) = f(\Lambda(A)^{-1}(z - a)).$
\item[{(F3)}]  Local commutativity: if the supports of $f$ and $g$ are space-like separated, then for any vector $v \in \cD$
$$[\phi_j(f)^{\diamond}, \phi_k(g)^{\diamond}]_{\mp} (v) = 0,$$
where $\diamond$ denotes the following possibilities: no $*$, one $*$, and both operators $\phi$ have a $*$.
\item[{(F4)}]  Vacuum; there exists a well defined vacuum state, i.e. a vector $\Omega \in \cH$, invariant with respect to the Poincar\'e group such that the following spectrum condition is satisfied :
the spectrum of the translation group $U(a, \jed)$ on $\Omega^{\bot}$ is contained in $V_m = \{ p ; (p,p) \geq m^2, p^0 > 0 \}$, $m>0$.
\end{enumerate}

\begin{remark}\label{poincare}
In an analysis of Poincar\'e-Lorentz transformations it is convenient to distinguish those related with \textbf{one frame}, $F$, and those which involve another frame, $F^{\prime}$, which moves with velocity $v$ relative to $F$.

In other words there are two basic types of Poincar\'e-Lorentz transformations:
\begin{enumerate}
\item transformations defined in terms of inertial frames with no relative motion, i.e. the frames are simply tilted.
In particular, there are rotations (but without continuous rotation) and translations.
\item transformations describing relative motion with constant (uniform) velocity and without rotations of the space coordinates. Such transformations \textbf{are called boosts.}
\end{enumerate}
It is worth pointing out that that the spectral conditions, mentioned in the point (L6) above, are relevant for the first type of 
Poincar\'e-Lorentz transformations. In particular, the spectral conditions mentioned there, are not applicable to generators of boosts.

The vital consequences of the observation just presented, will be described in the discussion on regularity conditions of fields operators in the next section. We will in particular show how natural regularity restrictions placed on the field operators in terms of the Hamiltonians of each of these types of transformations, each play a role in incorporating the field operators into the integration theory of local algebras.
\end{remark}

\section{Properties of field operators versus noncommutative integration.}\label{locint}

\subsection{The affiliation of field operators to local algebras} 

There are essentially two steps involved in showing that under mild restrictions, the field operators form a natural part of the integration theory of local algebras. The first step - which we review here - is to find conditions which ensure the affiliation of the field operators to local algebras. This involves a regularity restriction in terms of the Hamiltonian of the first set of transformations described in Remark \ref{poincare}. The second is to find conditions which ensure that the field operators are not only affiliated, but naturally embed into appropriate noncommutative function spaces associated with local algebras. We will show that this part of the scheme relies on regular behaviour with respect to the Hamiltonian of the second set of transformations described in Remark \ref{poincare}. The objective of showing how a field operator can be associated to a net of von Neumann algebras, was realised by by deep contributions from for example \cite{DSW}, \cite{buch1}, \cite{araki}. We pause to summarise those contributions. 

Let $\cP$ be a family of operators with a common dense domain of definition $\cD$ in a Hilbert space $\cH$ (cf. Wightman's rule presented above) such that if $\phi \in \cP$ then also $\phi^*|_{\cD} \equiv \phi^{\dagger} \in \cP$. The weak commutant $\cP^w$, of $\cP$ is defined as the set of all bounded operators $C$ on $\cH$ such that $(v,C\phi w) = (\phi^{\dagger}v, C w)$, for all $v,w \in \cD$.

For simplicity of our arguments we will restrict ourselves to one type of real scalar field $\phi$; i.e. $\phi(f)^*$ coincides with $\phi(\overline{f})$ on $\cD$. Furthermore, apart from the Wightman postulates we assume:

\begin{enumerate}
\item[{(A1)}]\label{I} $\cP(\cO^p_q)^w$ is an algebra for any double cone $\cO^p_q \equiv \{x; p-x \in V_+, x - q \in V_+ \}$, where $V_+ = \{ \rm{ positive \ timelike \ vectors} \ \}$.

\item[{(A2)}]\label{II} The vacuum vector $\Omega$ is cyclic for the union of $\cP(D^{\prime})^w$ over all double cones $D$, where $D^{\prime}$ is the causal complement of $D$.
\end{enumerate}

The following theorem is taken from \cite{araki} , but stems from results given in \cite{DSW}, \cite{buch1}.

\begin{theorem}[\cite{araki}, cf. \cite{buch1}\&\cite{DSW}]
\label{2.3}
Assume that both conditions (A1) and (A2) hold. For each double cone $D$, define
\begin{equation}
\qM(D) = \left( \cP(D)^w \right)^{\prime}.
\end{equation} 
Then $\qM(D)$ is a von Neumann algebra and the net $D \mapsto \qM(D)$ satisfies conditions (L1)-(L3) for local algebras (cf the first section). On defining $\qM$ to be the von Neumann algebra generated by $\cup_D \qM(D)$, the state on $\qM$ determined by $\Omega$ is then a pure vacuum state for which
\begin{itemize}
\item $\Omega$ is cyclic for each $\qM(D)$,
\item and each operator $\phi \in \cP(D)$ has a closed extension $\phi_e \subset \phi^{\dagger, *}$
which is affiliated with $\qM(D)$. (Here, $\phi_e \subset A$ means that the domain of $\phi_e$ is contained in the domain of $A$ and that $\phi_e = A$ on the domain of $\phi_e$.)
\end{itemize}
\end{theorem}
Theorem \ref{2.3} yields
\begin{corollary}
Field operators lead to operators affiliated to the von Neumann algebra $\qM(D)$. We remind that this property is the starting point for the definition of $\tau$-measurable operators.
\end{corollary}
Moreover, one has
\begin{remark}
There are sufficient conditions, motivated by physical requirements, for conditions (A1) and (A2) (given prior to Theorem \ref{2.3}) to hold, see \cite{araki}, \cite{DSW}, \cite{buch1}, \cite{BY}.
\end{remark}

As can be seen from Lemma 1.5 of \cite{BY} imposing a so-called \textit{generalized H-bound} condition (see below) ensures that (A1) holds. If one adds to this requirement the notion of central positivity, that would then ensure that the validity of Theorem \ref{2.3} (see \cite[Theorem 3.1]{BY}).

\begin{definition}
\label{regular}
 Let $\phi$ be a Wightman field and let $H$ denote its Hamiltonian. The field satisfies a \textit{generalized H-bound} if there exists a nonnegative number $\alpha <1$, such that $\phi(f)^{**} e^{- H^{\alpha}}$ is a bounded operator for all $f$.
\end{definition}

It is worth reiterating the fact pointed out in \cite{DSW}, that the physical significance of such conditions, is that they select models with slightly more regular high energy behaviour.

\subsection{Local algebras and the crossed product construction}

When provided with a von Neumann algebra with a faithful semifinite normal trace, integration theory is much simplified in that one is able to pass to the so-called algebra of $\tau$-measurable operators within which one may construct all the relevant noncommutative function spaces. However (see the comprehensive review \cite{Y}) \textit{the local algebras $\qM(D)$ are, under physically plausible assumptions, the same for all relativistic quantum field theories, namely they are isomorphic to the unique hyperfinite type $III_1$ factor}. But type $III$ factors are known not to have any nontrivial traces (see vol I and II of Takesaki \cite{Tak}). Haagerup's approach to noncommutative integration theory solves this problem by enlarging the ambient algebra by taking the crossed product $\mathfrak{M} \rtimes_{\sigma} \Rn$ (which will be defined below) of that algebra with the modular automorphism group. This enlarged algebra does indeed turn out to always admit a faithful semifinite normal trace, cf \cite{Tak}. Such a trace $\tau$, is a necessary tool for the definition of $\tau$-measurable operators, see \cite{Tak}, \cite{terp}, or \cite{nelson}, with the $L^p$-spaces of the original algebra realised as concrete spaces of operators within this algebra of $\tau$-measurable operators associated with the crossed product. Our first task here, will be to argue that the rudiments of this crossed product construction already appear in a natural way in the theory of local algebras. 

The Reeh-Schlieder property of vacuum states, see \cite{ReSch} (cf also Theorem 4.14 in \cite{araki}), shows that we can consider such vacuum states as those states whose GNS-vector $\Omega$ is cyclic and separating for the von Neumann algebra $\qM(D)$, where $D$ is a bounded region in the Minkowski space. Consequently, we get a very well-behaved Tomita-Takesaki theory in terms of these states. In particular, the modular action for the triple $(\qM(D), \Omega, \cH)$ exists and it will be denoted by $\sigma_t, \ t \in \Rn.$ This
modular action is therefore the natural object to use to construct such a crossed product.

Before proceeding with an analysis of the $\tau$-measurability criteria of field operators, we pause to describe $\cM \equiv \mathfrak{M} \rtimes_{\sigma} \Rn$ and its physical significance in some detail. To this end we will follow Haagerup's modification of Takesaki's construction as presented in Lemma 5.2 of \cite{uffe3}. This was developed in a series of three papers \cite{uffe1, uffe2, uffe3}. The bulk of the discussion below is taken from these papers, with occasional references from other sources where appropriate. For the sake of clarity we will extract only the very basic points of the exposition, without going to the point of sacrificing its essential content. Some modifications, which we made, are necessary in order to be able to follow the scheme of Quantum Field Theory as closely as possible. Recall that the basic ingredient of the operator algebraic approach to Quantum Field Theory is a net of local algebras having the properties described in the first section. A large part of our task is then to indicate how some of the subtleties of such algebras fit into the crossed product construction, rather than just presenting an abstract mathematical formalism.

Let $\Rn^4 \ni \cO \mapsto \qM(\cO)$ be a net of local observables. In general, $\qM(\cO)$ is type III. We fix a region $\cO_0$, so we can restrict ourselves to one von Neumann algebra $\qM(\cO_0) \equiv \qM$. Further there is a well defined vacuum state $\omega$ on 
$\qM$, having the properties described in the first section and then again at the beginning of the second section. Moreover on passing to the GNS representation if necessary, there is no loss of generality in assuming that $\qM$ is a von Neumann algebra acting on a Hilbert space $\cH$ with cyclic and separating vector $\Omega$. 
As $(\qM, \Omega, \cH)$ is then in the so-called standard form, there is a modular operator $\Delta$ inducing a modular automorphism group $\sigma_t(\cdot) = \Delta^{it} \cdot \Delta^{-it}$, $t \in \Rn$, of $\qM$.

Denote the Hilbert space of all square integrable $\cH$-valued functions on $\Rn$ by $L^2(\Rn, \cH)$. Define representations $\pi_{\sigma}$ of $\qM$, and $\lambda$ of $\Rn$, as follows:
\begin{equation}\label{cr1}
(\pi_{\sigma}(a)\xi)(t) = \sigma_{-t}(a)\xi(t), \quad a \in \qM, \ t \in \Rn, \xi \in L^2(\Rn, \cH) 
\end{equation}
and
\begin{equation}\label{cr2}
(\lambda(t)\xi)(s) = \xi(s - t), \quad s,t \in \Rn, \ \xi \in L^2(\Rn, \cH).
\end{equation}

\begin{definition}
The von Neumann algebra generated by $\pi_{\sigma}(\qM)$ and $\lambda(\Rn)$ on $L^2(\Rn, \cH)$, is called the crossed product of $\qM$ by $\sigma$, and is denoted by $\mathfrak{M} \rtimes_{\sigma} \Rn$.
\end{definition}

In the sequel we will write $\cM \equiv  \mathfrak{M} \rtimes_{\sigma} \Rn$ and also identify $\qM$ with $\pi_\sigma(\qM)$ to simplify the notation.

\begin{remark}
We know that whenever a generalised $H$-bound holds, any given field operator $\phi_e$ is affiliated with $\qM$ (that is taking the polar decomposition of 
$\phi_e$, $\phi_e = v |\phi_e|$, one has that $v$ and the spectral projections of $|\phi_e|$  are in $\qM$). But then $\phi_e$ will trivially also be a densely defined closed operator on $L^2(\Rn, \cH)$ for which $v$ and the spectral projections of $|\phi_e|$ are in $\cM$. Hence it must also be affiliated with $\cM$.
\end{remark}

The state $\omega$ on $\qM$ admits a so-called dual weight $\widetilde{\omega}$ on $\mathfrak{M} \rtimes_{\sigma} \Rn$, with the modular automorphism group on $\cM$ of this dual weight, being implemented by the maps $\lambda(t)$ ($t\in\mathbb{R})$) in the sense that $\tilde{\sigma}_{t}(a)=\lambda(t)a\lambda(t)^*$. (See the proof of \cite[Lemma 5.2]{uffe3}.) With $\sigma_t$ denoting the modular automorphism group on $\qM$ produced by $\omega$, it is not difficult to see that on embedding 
$\qM$ into $\cM$, we have by construction that $\sigma_s(a) = \lambda(s)a\lambda(s)^*$ for all $s\in\mathbb{R}$. So one of the benefits of passing to the crossed product and equipping it with the dual weight, is that we now have a modular group with an inner action. 

By the Stone-von Neumann theorem there exists a densely defined positive operator $h$ on $L^2(\Rn, \cH)$ such that $h^{it} =\lambda(t)$. But then by \cite[Theorem X.3.14]{Tak}, $\cM$ is not only semifinite, but $h$ must also be affiliated to $\cM$. Since $\hat{\omega}$ is in fact faithful, $h$ must be non-singular (dense-range and injective). We in fact have that $h^{it} = \hdelta^{it}$ where $\hdelta^{it}$ is the modular operator associated with the triple $(\cM, \hat{\omega},L^2(\Rn, \cH))$. (In the context of the left Hilbert algebra approach to crossed products, $\hdelta$ can be defined in a standard way as described in \cite[Lemma X.1.15]{Tak}.) The fact that $h^{it} = \hdelta^{it}$ can be verified by arguing as in the uniqueness part of the proof of \cite[Theorem VIII.1.2]{Tak} (see page 93).

To emphasise the connection of the above crossed product construction with physical criteria, we will follow some arguments given in Chapter V.4 in \cite{haag}. Let $W$ denote the wedge in Minkowski space $\mathbb{M}$.
\begin{equation}
W = \{r \in \mathbb{M}: r^1 > |r^0|, \ r^2,r^3 \quad \rm{arbitrary} \}
\end{equation}
and $U(\Lambda(s))$ ($U(r)$) the unitary operators implementing the boosts $\Lambda(s)$ (the spacetime translations respectively). The corresponding generators will be denoted by $K$ and $P_{\mu}$, i.e.
\begin{equation}
U(\Lambda(s)) = e^{iKs}, \quad U(r) = e^{iP_{\mu}}r^{\mu}.
\end{equation}

Finally, let $\Theta$ stand for CPT-operator (cf Chapter II in \cite{haag}) and $R_1(\pi)$ denote a special rotation through the angle of $\pi$ around the $1$-axis. Bisognano and Wichmann proved, see \cite{BW1}, \cite{BW2}

\begin{theorem}[\cite{BW1}, \cite{BW2}]\label{BW}
Let $J_W$, $\Delta_W$ denote the modular conjugation and the modular operator for the pair $(\mathfrak{M}(W), \Omega)$.
Then
\begin{equation}
J_W = \Theta U(R_1(\pi)), \quad \Delta_W = e^{-2\pi K}.
\end{equation}
Moreover, \textit{the modular automorphisms $\sigma_t$, $t \in \Rn$, act geometrically as the boosts!}
\end{theorem}

We pause to note that although the Bisognano-Wichmann result is a theorem about $\qM(W)$, it is also a theorem about the geometry of the canonical modular group. Hence with $(\sigma_t)$ denoting the modular group, the fact that inside the crossed product $\sigma_t$ may be realised by $\sigma_t(a)=\lambda(t)a\lambda(t)^*$, provides an elegant way of connecting their result with the crossed product approach.

\begin{remark}
As far as  $\cM \equiv  \mathfrak{M} \rtimes_{\sigma} \Rn$ is concerned, one could therefore write 
\begin{equation}
h = \hdelta = e^{-H_{equilibrium}} \equiv e^{-K},
\end{equation}
where, using ``KMS'' ideology, $K$ can be interpreted as the Hamiltonian describing the equilibrium dynamics for some fixed (non-moving) frame. For this (modular) dynamics, $\omega$ is moreover an equilibrium (KMS)-state. It is important to remember that the spectrum $\sigma(K)$ of $K$, in general, satisfies
$$\sigma(K) = \Rn.$$
Note that this property is a characteristic feature of modular groups of type $III_1$ factors! To fully appreciate this observation we remind the reader that $\widetilde{\sigma}_t(A)=\sigma_t(a)$ for all $a\in\qM$, where $\widetilde{\sigma}_t$ (respectively $\sigma_t$) is the modular group produced by the dual weight $\widehat{\omega}$ (respectively the weight $\omega$). (Compare the earlier discussion following the introduction of the operator-valued weight.)

On the other hand, it is worth pointing out that, in general, $K$ is not equal to $H$. Thus, the spectral conditions assumed for $H$, are not applicable to $K$!
\end{remark}

\subsection{Field operators and $\tau$-measurability}

For the sake of completeness we remind the reader of the concept of $\tau$-measurability.

\begin{definition}
\label{measure}
Let $\cM$ be a semifinite von Neumann algebra acting on a Hilbert space $\cH$, and equipped with a faithful normal semifinite trace $\tau$. A closed and densely defined linear operator $a$ on $\cH$ is called $\tau -measurable$ if $a$ is affiliated with $\cM$ (that is, $au=ua$ for all unitary operators in the commutant $\cM^{\prime}$ of $\cM$) and there exists a $\lambda_0 \in (0,\infty)$ such that $\tau(E^{|a|}(\lambda_0, \infty)) < \infty$ (where $E^{|a|}$ denotes the spectral measure of the selfadjoint operator $|a|$). The space of all $\tau$-measurable operators is denoted by $\widetilde{\cM}$ and can be equipped with a topology of convergence in measure, with respect to which it proves to be a complete metrisable *-algebra. Details may be found in \cite[\S IX.2]{Tak}. This *-algebra is the noncommutative analogue of the completion of $L^\infty(X,\Sigma,\nu)$ with respect to the topology of convergence in measure.
\end{definition}  

We pause to describe how quantum $L^p$-spaces spaces are realised within the algebra $\widetilde{\cM}$.  One may define a dual action of $\mathbb{R}$ on $\cM$ in the form of a one-parameter group of automorphisms $(\theta_s)$ by means of the prescription 
\begin{equation}\label{dual} \theta_s(a)=a, \quad \theta_s(\lambda(t))=e^{-ist}\lambda(t) \mbox{ for all } a\in \qM\mbox{ and } s,t\in\mathbb{R}.\end{equation} 
(We remind the reader that we have chosen to identify the copy $\pi_\sigma(\mathfrak{M})$ of $\mathfrak{M}$ inside $\cM$, with $\mathfrak{M}$ itself in order to simplify notation.) This dual action in fact extends to an action on $\widetilde{\cM}$. Using this dual action, Haagerup showed that  $\qM=\{a\in\widetilde{\cM}:\theta_s(a)=a\mbox{ for all } s\in\mathbb{R}\}$ and that $\qM_*\equiv\{a\in\widetilde{\cM}:\theta_s(a)=e^{-s}a\mbox{ for all }s\in\mathbb{R}\}$. So if $\qM$ and $\qM_*$ respectively represent the quantum spaces $L^\infty(\qM)$ and $L^1(\qM)$, it makes sense to suggest that the $L^p(\qM)$ spaces $1\leq p\leq\infty$ may well be represented by the scale of spaces  $L^p(\qM)=\{a\in\widetilde{\cM}:\theta_s(a)=e^{-s/p}a\mbox{ for all }s\in\mathbb{R}\}$. Haagerup's triumph was in showing that these are all Banach spaces with respect to the subspace topology inherited from the topology of convergence in measure on $\widetilde{\cM}$, and that there exists a tracial functional $tr$ on $L^1(\qM)$ (different from $\tau$) in terms of which one may define a dual action of $L^{p'}$ on $L^p$, which can then be used to develop a theory which reproduces much of the classical theory of $L^p$-spaces with a remarkable degree of faithfulness. A consequence of this theory is the fact that in the case where $\omega$ is a state (as in the present case), we have that $h\in L^1(\qM)$ with $\omega(a)=tr(ha)=tr(ah)=tr(h^{1/2}ah^{1/2})$ for all $a\in \qM$. 

We now explore natural ways in which to assign $\tau$-measurability criteria to the field operators.  Let us now apply the above ideas to the action of the field operators on $\qM(D)$ where $D$ is some double cone. At this point it is very seductive to suggest that one could simply require the field operators themselves to be measurable. But there is a problem with that. It can be shown that the action of each $\theta_s$ extends to affiliated operators, and hence also to the field 
operators. Each positive operator affiliated to $\mathfrak{M}(D)$ can of course be written as a pointwise increasing limit of positive elements of $\mathfrak{M}(D)$. But we know from \cite[Lemma 3.6]{uffe-DW1}, that the elements of $\qM^+(D)$ are fixed points of this dual action. In other words we must then have that $\theta_s(|\phi(f)|)=|\phi(f)|$ for all $s$. So if $\phi(f)$ was in fact $\tau$-measurable, it would follow from Proposition II.10 of 
\cite{terp} that in the context of $\qM(D)$, the action of these operators would be bounded! But this cannot be. It is therefore clear 
that when imposing a regularity assumption related to $\tau$-measurability, we need to take care that we impose a regularity restriction 
which allows for unboundedness of the field operators. In search of such a criterion we turn to \cite{ML} for clues. In that paper a strong case was made that regular observables find their home in the Orlicz space $L^{\cosh-1}(\mathfrak{M}(D))$ (defined below). The significance of this space is that for all non-discrete von Neumann algebras (as is the case with $\mathfrak{M}(D)$), it ``contains'' 
$\mathfrak{M}$ properly, and hence allows for unboundedness of observables. (What we at this point mean by containment should be explained. When in the tracial case we speak of containment, we mean standard set theoretic containment. However in the crossed product paradigm this translates to the statement that there exists a canonical adjoint preserving embedding of $\mathfrak{M}$ into $L^{\cosh-1}(\mathfrak{M}(D))$.) So a natural way in which we can assign $\tau$-measurability criteria to the field operators, is to ask that they canonically embed into a physically appropriate quantum function space associated with $\mathfrak{M}(D)$. We will shortly comment on the propriety of $L^{\cosh-1}(\mathfrak{M}(D))$ as a home for the field operators. We first pause to define Orlicz spaces, and describe how their quantum analogues are constructed. In order to introduce Orlicz spaces, we first need to acquaint readers with the concept of a Young's function.  

\begin{definition}
A convex function $\Psi : [0, \infty) \to [0, \infty]$ is called a Young's function if 
\begin{itemize}
\item  $\Psi(0) = 0$ and $\lim_{u \to \infty} \Psi(u) = \infty$;
\item it is neither identically zero nor infinite valued on all of 
$(0, \infty)$,
\item and is left continuous at $b_\Psi = \sup\{u > 0 : \Psi(u) < 
\infty\}$.
\end{itemize}
\end{definition}

Young's functions come in complementary pairs, in that to each Young's function $\Psi$, we may associate the complementary function $\Psi^*$ defined by $\Psi^*(t)=\sup_{s>0}st-\Psi(s)$. By definition of $\Psi^*$, the pair $(\Psi,\Psi^*)$ then satisfy the Hausdorff-Young inequality $st\leq \Psi(s)+\Psi^*(t)$ for all $s,t\geq 0$. To each Young's function $\Psi$, we may associate a corresponding Orlicz space. With $L^0$ denoting the space of all measurable functions on some $\sigma$-finite measure space $(\Omega, \Sigma, m)$, the Orlicz space associated with a given Young's function $\Psi$, may be defined by the prescription:

\begin{definition}\label{orliczelement} $f\in L^0$ belongs to $L^{\Psi} \Leftrightarrow \int\Psi(\lambda |f|)\,d\nu<\infty$ for some $\lambda = \lambda(f) > 0.$
\end{definition}

These spaces admit of several natural norms, but ultimately these norms turn out to be equivalent. (See \cite[\S 4.8]{BS}.) For the sake of completeness we mention the two most important norms.

\begin{itemize}
\item \textbf{Luxemburg-Nakano norm}: $\|f\|_\Psi = \inf\{\lambda > 0 : \|\Psi(|f|/\lambda)\|_1 \leq 1\}.$

\item \textbf{Orlicz norm}: $\|f\|^O_\Psi = \sup\{|\textstyle{\int_\Omega} fg\, dm| : g\in L^{\Psi^*}, \|g\|_{\Psi^*}\leq 1\}.$
\end{itemize}

Given a Young's function, we will follow the convention of writing $L^\Psi$ when the Luxemburg-Nakano norm is in view, and $L_\Psi$ when the Orlicz norm is used. To be able to define Orlicz spaces for type III algebras, we need the concept of a fundamental function, specifically fundamental functions of the Orlicz spaces $\{L^\Psi(0,\infty), L_\Psi(0,\infty)\}$. For each such space, say $X$, the associated \emph{fundamental function} is defined by the prescription
$$\varphi_X(t)=\|\chi_E\|_X \mbox{ where } \nu(E)=t.$$
In the case $X=L^\Psi(0,\infty)$ we will write $\varphi_\Psi$ for the fundamental function, and in the case $X=L_\Psi(0,\infty)$ write $\widetilde{\varphi}_\Psi$. These functions turn out to be so-called ``quasi-concave functions'' on $[0,\infty)$ (see \cite[\S 2.5]{BS} for a detailed discussion of such functions). For these fundamental functions the following facts are known to hold (see \cite[Theorem 2.5.2, Corollary 4.8.15 \& Lemma 4.8.17]{BS}):

\begin{equation}\label{fundeqn} \varphi_\Psi(t)=\frac{1}{\Psi^{-1}(1/t)} \mbox{ and }\varphi_\Psi(t)\widetilde{\varphi}_{\Psi^*}(t)=t\mbox{ for all }t.\end{equation}

Given $\Psi$, the Orlicz space $L^\Psi(\qM)$ associated with $\qM$ may then be defined by means of the following prescription (see 
\cite[Lemma 5.11]{LM2}): \newline
Let $v_s=\varphi_{\Psi}(e^{-s}h)\varphi_{\Psi}(h)^{-1}$. Then  
$$L^\Psi(\qM) =\{a\in \widetilde{\cM}: \theta_s(a)=v_s^{1/2}av_s^{1/2} \mbox{ for all } s\leq 0\}.$$To define $L_\Psi(\qM)$ one simply uses $\widetilde{\varphi}_{\Psi}$ in the definition instead of $\varphi_{\Psi}$. We point out that in the case where $\Psi(t)=t^p$, this prescription yields exactly Haagerup's definition of $L^p$-spaces. At this point it is worth noting that the operator $v_s$ is always bounded with norm between 1 and $e^{-s}$. This can be proven by noting that the facts that $t\to \varphi_{\Psi}(t)$ is increasing and $t\to\frac{\varphi_{\Psi}(t)}{t}$ decreasing (see \cite[Corollary II.5.3]{BS}), ensures that for any $t>0$, $\frac{\varphi_{\Psi}(e^{-s}t)}{\varphi_{\Psi}(t)}$ lies between 1 and $e^{-s}$. The continuous functional calculus does the rest.

Coming back to field operators and their action in the context of $\qM(D)$, the question of precisely which regularity criteria to impose now arises. We remind the reader that the self-adjoint elements of the space $L^{\cosh-1}(\mathfrak{M}(D))$ represents the space of regular observables \cite{ML, LM}. As was shown in \cite{LM} one of the properties of such regular observables, is that they have all moments finite. In some sense the field operators exhibit similar behaviour. We note the following:

\begin{observation}
For any field operator $\phi(f)$ satisfying Wightman's postulates, one has
\begin{equation}\label{fmoment}
(v,\phi^n(f)w) \in \Cn,
\end{equation}
for any $v,w \in \cD\subset\cH$, and any $n \in \Nn$. In other words, the number $(v,\phi^n(f)w)$ is finite for any $n \in \Nn$.
\end{observation}

\begin{remark}
For states $\omega_x(\cdot) \equiv (x, \cdot x)$ with $x \in \cD\subset\cH$, field operators in QFT enjoy the property of having all moments finite.
\end{remark}

The above facts suggest that from a physical point of view, the space $L^{\cosh-1}(\mathfrak{M}(D))$, ought to be a natural home for the field operators. If then the field operators are to be \emph{regular} in the sense of \cite{ML, LM}, a natural restriction to place on them would be to require them to embed into the space  $L^{\cosh-1}(\mathfrak{M}(D))$. This leads us to the following definition

\begin{definition}
A field operator $\phi(f)$ affiliated to $\mathfrak{M}(D)$ is said to satisfy an \emph{$L^{\cosh-1}$ regularity restriction} if the strong product $\varphi_{\cosh-1}(h)^{1/2}\phi(f)_e\varphi_{\cosh-1}(h)^{1/2}$ (where $\phi(f)_e$ is the minimal closed extension of $\phi(f)$) is a closable operator for which the closure is $\tau$-measurable, i.e. the closure is an element of the space $\widetilde{\cM}$. For the sake of simplicity of notation we will hereafter simply write $\varphi_{\cosh-1}(h)^{1/2}\phi(f)\varphi_{\cosh-1}(h)^{1/2}$ instead of $\varphi_{\cosh-1}(h)^{1/2}\phi(f)_e\varphi_{\cosh-1}(h)^{1/2}$.)
\end{definition}

\begin{remark} Given that each of the $\phi(f)$'s are fixed points of the action of $(\theta_s)$, it is a simple matter to 
verify that an $L^{\cosh-1}$ regularity restriction on some $\phi(f)$, ensures that the closure of $\varphi_{\cosh-1}(h)^{1/2}\phi(f)\varphi_{\cosh-1}(h)^{1/2}$ satisfies the membership criteria for $L^{\cosh-1}(\mathfrak{M}(D))$ described above. To see this observe that for $v_s=\varphi_{\cosh-1}(e^{-s}h)\varphi_{\cosh-1}(h)^{-1}$, we clearly have that $\theta_s(\varphi_{\cosh-1}(h)^{1/2}\phi(f)\varphi_{\cosh-1}(h)^{1/2}) = \varphi_{\cosh-1}(e^{-s}h)^{1/2}\phi(f)\varphi_{\cosh-1}(e^{-s}sh)^{1/2}= v_s^{1/2}(\varphi_{\cosh-1}(h)^{1/2}\phi(f)\varphi_{\cosh-1}(h)^{1/2})v_s^{1/2}$ for each $s\in \mathbb{R}$.
\end{remark}

We re-emphasise that the $L^{\cosh-1}$ regularity restriction allows room for unboundedness of the field operators! This restriction clearly has the same flavour as the generalised $H$-boundedness restriction on the field operators $\pi_{\sigma}(\phi_e)$ mentioned earlier, which implied their affiliation to the local algebras $\mathfrak{M}(D)$. The differences here is that we used the Hamiltonian for boosts rather than $H$, that we have a ``symmetrised'' product, and that we ask for this product to be merely $\tau$-measurable rather than bounded. 

\subsubsection{Generalized $H$-bounds and $L^{\cosh-1}$ regularity}
 
We now show that under very natural criteria, a large number of fields satisfying a generalised $H$-bound condition will automatically satisfy an $L^{\cosh-1}$ regularity restriction. The crucial assumption here is that the vacuum vector is an analytic vector for the Hamiltonian $H$. Note that this assumption of analyticity is intimately related to the validity of the Reeh-Schlieder theorem, which is a key ingredient in Quantum Field Theoretic modelling. To see this, the reader may refer to \cite[1.3.1-1.3.3]{hor}.
 
\begin{theorem}\label{Lcosh-criteria}
Let a quantum field be given which satisfies a generalised $H$(-power)-bound condition, in the sense that each $\phi(f)$ is a densely defined hermitian operator with $\pm\phi(f)\leq K(f) (\I+H)^n$ for some positive integer $n$ and $K(f)>0$, where $K(f)$ may depend on $f$. Let $e_\omega$ be the support projection of the state $\omega=\langle \cdot \Omega_\omega,\Omega_\omega\rangle$. Then field operators satisfy the compression $e_\omega\phi(f)e_\omega$ of the field satisfies an $L^{\cosh-1}$ regularity restriction for the compressed algebra $e_\omega\qM e_\omega$ in the sense that $\varphi_{\cosh - 1}(g_\omega)^{1/2}e_\omega\phi(f)e_\omega\varphi_{\cosh - 1}(g_\omega)^{1/2}$ is pre-closed with the closure being a $\tau$-measurable operator affiliated to $e_\omega\qM e_\omega\rtimes_\omega\mathbb{R}$. Here $g_\omega$ is the density of the dual weight $\widetilde{\omega}$ on $e_\omega\qM e_\omega\rtimes_\omega\mathbb{R}$.
\end{theorem}

We point out that in one of the first concrete examples of a field satisfying a generalised $H$-bound, Driessler and Fr\"ohlich \cite{DF} proved that for the so-called Osterwalder-Schrader reconstruction, a generalised $H$-bound pertains in the sense that $\pm \phi(f)\leq\|f\|(H+\I)$ for some continuous norm on the space of test functions. In general such a linear bound is too restrictive with a much larger number of examples satisfying a condition of the form described above. (See the discussion after Definition 5.1 of \cite{DSW}.)

\begin{remark} In Algebraic Quantum Field Theory the net of algebras $\cO \mapsto \qM(\cO)$ constitutes the intrinsic mathematical description of the theory, cf Section 1. Hence for the vacuum state $\omega$, the essential information is provided by
\begin{equation}
\label{net}
\qM(\cO)  \ni a \mapsto \omega(a).
\end{equation}
We note that as $\omega(e_\omega a e_\omega) = \omega(a)$ for $a \in \cup_{\cO} \qM(\cO)$,  the net \ref{net} can be used to great effect in the compressed algebra setting. In other words, the passage to the compression determined by $e_\omega$ does not spoil the basic features of the theory. The advantage of using the compression $e_\omega \cdot e_\omega$ will be explained in Section 
\ref{concludesection2}.
\end{remark}

\begin{proof}[Proof of Theorem \ref{Lcosh-criteria}]
Throughout the proof, we shall write $\qM$ for $\qM(\mathbb{M})$. We firstly note that by \cite[Theorem III.4.2]{Tak}, we have that $e_\omega\in\qM$. The further fact that the unitary group of which $H$ is the generator, is contained in $\qM$, ensures that its generator, namely $H$, is affiliated to $\qM$, and hence also to $\cM$. We shall denote this group by $U(a)$ ($a\in\mathbb{M}$) since only translation is in view. 

Select a faithful normal semifinite weight $\delta$ on $(\I-e_\omega)\qM(\mathbb{M})(\I-e_\omega)$. The weight $\delta$ extends to a normal semifinite weight on $\qM(\mathbb{M})$ in the obvious way. We now define a weight $\nu$ on $\qM(\mathbb{M})$ by the prescription 
$$a\to \omega(e_\omega ae_\omega)+\delta((\I-e_\omega)a(\I-e_\omega))\qquad a\in \qM(\mathbb{M}).$$It is an exercise to see that $\nu$ 
is a faithful normal semifinite weight on $\qM(\mathbb{M})$. (To see the faithfulness, notice that for $a\geq 0$, we have 
that $\nu(a)=0\Rightarrow e_\omega ae_\omega=(\I-e_\omega)a(\I-e_\omega)=0\Rightarrow a^{1/2}e_\omega=a^{1/2}(\I-e_\omega)=0 \Rightarrow a^{1/2}=0\Rightarrow a=0$.) Now pass to the crossed product $\qM(\mathbb{M})\rtimes_\nu\mathbb{R}$. 
Given a normal semifinite weight $\gamma$ on $\qM(\mathbb{M})$, we will write $h_\gamma$ for the density of the dual weight. From the notes of Terp we know that $h_\nu=h_\omega \dotplus h_\delta$ and that $\mathrm{supp}(h_\omega)=\mathrm{supp}(\omega)=e_\omega$ and 
$\mathrm{supp}(h_\delta)=\mathrm{supp}(\delta)=\I-e_\omega$. Thus $h_\omega h_\delta=0$ which ensures that 
$h_\nu^{it}=h_\omega^{it} + h_\delta^{it}$ for all $t\in \mathbb{R}$. But then $\sigma^\nu_t(e_\omega)=h_\nu^{it}e_\omega h_\nu^{-it} =(h_\omega^{it} + h_\delta^{it})e_\omega(h_\omega^{-it} + h_\delta^{-it})=h_\omega^{it}e_\omega h_\omega^{-it}=e_\omega$ for all 
$t\in \mathbb{R}$. 

Now recall that the canonical normal state $\omega$ corresponds to the vacuum state. But the vacuum state is invariant under the action of the automorphisms $\alpha_a$ implemented by translation at the level of $\mathbb{M}$ \cite[Theorem 4.5]{araki}. This means 
that for any $f\in\qM$ and any $a\in\mathbb{M}$, we must have that $\omega(f)=\omega(\alpha_a(f))=\omega(U(a)fU(a)^*)$. Next observe 
that in the language of Haagerup $L^p$-spaces, this fact corresponds to the claim that $tr(h_\omega f)=tr(h_\omega U(a)fU(a)^*)= tr(U(a)^*h_\omega U(a)f)$ for all $f\in \qM$, where $h_\omega$ is the density of the dual weight $\widetilde{\omega}$ on 
$\qM\rtimes_\nu\mathbb{R}$. (See the discussion following Definition \ref{measure}.) The well-developed duality theory of Haagerup 
$L^p$-spaces, combined with the fact that both $h_\omega$ and $U(a)^*h_\omega U(a)$ belong to $L^1(\qM)$, now ensures that $h_\omega=
U(a)^*h_\omega U(a)$, or equivalently that $U(a)h_\omega=h_\omega U(a)$ for each $a\in\mathbb{M}$. When expressed at the level of the 
generator $H$ of the group $U(a)$ ($a\in\mathbb{M}$), this corresponds to the claim that $H$ and $h_\omega$ are commuting affiliated operators 
of $\cM$.

By postulate (F4), the vacuum vector is in the kernel of the energy operator $H$, and is hence automatically analytic for all powers $H^k$ of $H$. We shall need this fact shortly. We also remind the reader that in the context of Haagerup $L^p$-spaces, the standard form of a von Neumann algebra $\qM$, may be realised by representing $\qM$ as left-multiplication operators on $L^2(\qM)$ - see \cite[Theorem II.36]{terp}. 

Using this theorem, the vector $\Omega$ implementing the canonical faithful normal state $\omega$, may be identified with $h_\omega^{1/2}$. Saying that $h_\omega^{1/2}$ is in the domain of say $H^k$, in this setting means that $H^kh_\omega^{1/2}$ is $\tau$-pre-measurable and that the extension (for which we use the same notation) is in $L^2(\qM)$. The analyticity assumption regarding the vacuum vector ensures that $\sum_{m=1}^\infty\frac{1}{m!}t^mH^{mk}\Omega$ converges absolutely, which in turn can only be true if each $H^{mk}h_\omega^{1/2}$ extends uniquely to an element of $L^2(\qM)$, with the series $\sum_{m=0}^\infty\frac{1}{m!}t^mH^{mk}h_\omega^{1/2}$ of such extensions, converging absolutely to an element in $L^2(\qM)$. We may then use the comparison test for series to conclude that $\sum_{m=0}^\infty\frac{1}{(2m)!}\|t^{2m}H^{2mk}h_\omega^{1/2}\|_2$ converges. The absolute convergence of the series  $\sum_{m=0}^\infty\frac{1}{(2m)!}t^{2m}H^{2mk}h_\omega^{1/2}$ in $L^2(\qM)$, ensures that it must correspond to a well defined element of 
$L^2(\qM)$. But this series is just a Maclaurin series representation of $(\cosh-1)(tH^k)h_\omega^{1/2}$. So $(\cosh-1)(tH^k)h_\omega^{1/2}\in L^2(\qM)$. Now recall that the fact that $\omega$ is a state, ensures that $h_\omega\in L^1(\qM)$, or equivalently that $h_\omega^{1/2}\in L^2(\qM)$. We therefore have that $h_\omega^{1/2}(\cosh-1)(tH^k)h_\omega^{1/2}\in L^1(\qM)$ for every $1\leq k\leq n$.

The fact that $H$ and $h_\omega$ are commuting affiliated operators now comes into play. On arguing as in the proof of \cite[Theorem 2.2]{L}, we may show that $$\tau(\chi_{(1,\infty)}(h_\omega^{1/2}(\cosh-1)(tH^k)h_\omega^{1/2})) = \tau(\chi_{(1,\infty)}(\varphi_{\cosh-1}(h_\omega)^{1/2}(tH^k)\varphi_{\cosh-1}(h_\omega)^{1/2})).$$(Even though $\qM$ is not semifinite, the argument in the proof of \cite[Theorem 2.2]{L} will go through if we apply it to the pair $(H,h_\omega)$ instead of the pair $(a\otimes 1, \I\otimes e^t)$.) By \cite[Lemma 1.7]{FK} we have that $\tau(\chi_{(1,\infty)}(h_\omega^{1/2}(\cosh-1)(tH^k)h_\omega^{1/2}))=\|h_\omega^{1/2}(\cosh-1)(tH^k)h_\omega^{1/2}\|_1<\infty$, and hence that $$\tau(\chi_{(1,\infty)}(\varphi_{\cosh-1}(h_\omega)^{1/2}(tH^k)\varphi_{\cosh-1}(h_\omega)^{1/2}))<\infty.$$Thus by definition, $\varphi_{\cosh-1}(h_\omega)^{1/2}(tH^k)\varphi_{\cosh-1}(h_\omega)^{1/2}$ is $\tau$-measurable. Since this holds for every $1\leq k \leq n$, it follows that $\varphi_{\cosh-1}(h_\omega)^{1/2}(\I+H)^n\varphi_{\cosh-1}(h_\omega)^{1/2}$ is $\tau$-measurable.

Recall that for any test function $f$ we have that $\pm\phi(f)_e \leq K(f) (\I+H)^n$. It is therefore clear that
\begin{eqnarray*}
0 &\leq& \varphi_{\cosh-1}(h_\omega)^{1/2}(K(f)(\I+H)^n+\phi(f)_e)\varphi_{\cosh-1}(h_\omega)^{1/2}\\
&\leq& 2K(f)\varphi_{\cosh-1}(h_\omega)^{1/2}(\I+H)^n\varphi_{\cosh-1}(h_\omega)^{1/2}.
\end{eqnarray*}
We may now use this inequality and the fact that $\mathrm{supp}(h_\omega)=\mathrm{supp}(\omega)=e_\omega$, to conclude that each $\varphi_{\cosh-1}(h_\omega)^{1/2}(K(f)(\I+H)^n+\phi(f)_e)\varphi_{\cosh-1}(h_\omega)^{1/2}$ is $\tau$-measurable, and hence also each \newline $\varphi_{\cosh-1}(h_\omega)^{1/2}\phi(f)_e\varphi_{\cosh-1}(h_\omega)^{1/2}=\varphi_{\cosh-1}(h_\omega)^{1/2}e_\omega\phi(f)_ee_\omega\varphi_{\cosh-1}(h_\omega)^{1/2}$.

If now we are able to show that in all respects the subalgebra of $\qM \rtimes_\nu\mathbb{R}$ generated by $e_\omega\qM e_\omega$ and the compressed left-translation operators $e_\omega\lambda(t)e_\omega$ is an exact copy of $e_\omega\qM e_\omega\rtimes_\omega\mathbb{R}$, it will follow from the above fact that $\varphi_{\cosh - 1}(g_\omega)^{1/2}e_\omega\phi(f)_ee_\omega\varphi_{\cosh - 1}(g_\omega)^{1/2}$ is a $\tau$-measurable operator affiliated to $e_\omega\qM e_\omega\rtimes_\omega\mathbb{R}$. 

We proceed to investigate the algebra generated by $e_\omega\qM e_\omega$ and the compressed left-translation operators $e_\omega\lambda(t)e_\omega$. For the sake of clarity we will write $\lambda_\nu(t)$ and $\lambda_\omega(t)$ for the shift operators corresponding to respectively $\qM \rtimes_\nu\mathbb{R}$ and $e_\omega\qM e_\omega\rtimes_\omega\mathbb{R}$. This part of the proof requires a more delicate treatment of the crossed product. We have up to now been identifying the algebra $\qM$ with the copy thereof $\pi_\nu(\qM)$ inside $\qM \rtimes_\nu\mathbb{R}$. We remind the reader that the crossed product $\qM\rtimes_\nu \mathbb{R}$ is properly a subalgebra of the von Neumann algebra tensor product $\qM\otimes B(L^2(\mathbb{R}))$ generated by the copy $\pi_\nu(\qM)$ of $\qM$, and the shift operators $\lambda_\nu(t)$. In particular each $\lambda_\nu(t)$ ($t\in \mathbb{R}$) is of the form $\I\otimes \ell_t$, where $\ell_t$ is the left-translation operator defined on $L^2(\mathbb{R})$ by $\ell_t(f)(s)=f(s-t)$. Similarly $(e_\omega\qM e_\omega)\rtimes_\nu \mathbb{R}$ is a subalgebra of the von Neumann algebra tensor product $(e_\omega\qM e_\omega)\otimes B(L^2(\mathbb{R}))$. 

It is useful to note that the subalgebra of $\qM \rtimes_\nu\mathbb{R}$ generated by $e_\omega\qM e_\omega$, and the compressed shift operators $e_\omega\lambda_\nu(t)e_\omega$ corresponds exactly to the compression $e_\omega(\qM \rtimes_\nu\mathbb{R})e_\omega$. To see this we note that for any $a\in\qM$ and any $t\in\mathbb{R}$, we have that $a\lambda_\nu(t)=\lambda_\nu(t)\lambda_\nu(-t)a\lambda_\nu(t)=\lambda_\nu(t)\sigma_{-t}^\nu(a)$. Using this fact it is clear that any finite algebraic combination of elements of $\qM$ and shift operators $\lambda_\nu(t)$, may be written as a linear combination of terms of the form $\lambda_\nu(t)a$. Thus $\mathrm{span}\{\lambda_\nu(t)a:a\in\qM, t\in\mathbb{R}\}$ is weak* dense in $\qM \rtimes_\nu\mathbb{R}$. We have already seen that $\lambda_\nu(t)=h_\nu^{it}$ commutes with $e_\omega$. Using that fact, it now follows that $e_\omega(\lambda_\nu(t)a)e_\omega=(e_\omega\lambda_\nu(t)e_\omega)(e_\omega ae_\omega)$. If we consider this alongside the weak* density noted above, the claim now clearly follows 

The fact that as von Neumann algebras the subalgebra of $\qM \rtimes_\nu\mathbb{R}$ generated by $e_\omega\qM e_\omega$, and the compressed left-translation operators $e_\omega\lambda_\nu(t)e_\omega$, and $(e_\omega\qM e_\omega)\rtimes_\nu \mathbb{R}$ are copies of each other, follows from the fact that the compression $(e_\omega\otimes\I)[\qM\otimes B(L^2(\mathbb{R}))](e_\omega\otimes\I)$ is an exact copy of $(e_\omega\qM e_\omega)\otimes B(L^2(\mathbb{R}))$. (This can be seen by using the relation between crossed products and  tensor products given above.) We proceed verify the finer details that will identify these two objects in a crossed product sense, rather than just von Neumann algebras. 

\begin{itemize}
\item Firstly note that the prescription $(e_\omega\xi)(s)= e_\omega(\xi(s))$ defines a map which trivially maps $L^2(\mathbb{R},\mathcal{H})$ onto $L^2(\mathbb{R},e_\omega\mathcal{H})$. More properly recall that $L^2(\mathbb{R},\mathcal{H})\equiv \mathcal{H}\otimes L^2(\mathbb{R})$ and $L^2(\mathbb{R},e_\omega\mathcal{H})\equiv e_\omega\mathcal{H}\otimes L^2(\mathbb{R})$. So in this tensor product notation, this map is nothing but $e_\omega\otimes\I$.
\item We next claim that $\pi_\nu(e_\omega ae_\omega)\xi=\pi_\omega(e_\omega ae_\omega)(e_\omega\xi)$ for each $\xi\in L^2(\mathbb{R},\mathcal{H})$ and each $a\in \qM$: To see this note that 
\begin{eqnarray*}
\pi_\nu(e_\omega ae_\omega)\xi(s) &=& \sigma_{-s}^\nu(e_\omega ae_\omega)\xi(s)\\ 
&=& h_\nu^{-is}(e_\omega ae_\omega)h_\nu^{is}\xi(s)\\
&=& h_\omega^{-is}(e_\omega ae_\omega)h_\omega^{is}\xi(s)\\ 
&=& h_\omega^{-is}(e_\omega ae_\omega)h_\omega^{is}e_\omega\xi(s)\\
&=& \sigma_{-s}^\omega(e_\omega ae_\omega)e_\omega\xi(s)\\
&=& \pi_\omega(e_\omega ae_\omega)(e_\omega\xi)(s)
\end{eqnarray*}
\item We next claim that $[e_\omega,\lambda_\nu(t)]=0$ with $e_\omega\lambda_\nu(t)$ corresponding to $\lambda_\omega(t)$ on $L^2(\mathbb{R},e_\omega\mathcal{H})$: The fact that $[e_\omega,\lambda_\nu(t)]=0$, follows from the fact that $e_\omega$ commutes with $h_\nu$ combined with the fact that $\lambda_\nu(t)=h_\nu^{it}$. The second claim follows from the easily verifiable fact that the action of $\lambda_\nu(t)$ on $L^2(\mathbb{R},e_\omega\mathcal{H})$ agrees with the action of $\lambda_\omega(t)$ on $L^2(\mathbb{R},e_\omega\mathcal{H})$.
\item We finally claim that for any $a\in\qM$ and any $s,t\in\mathbb{R}$, we have that $\theta_s^{(\nu)}(e_\omega ae_\omega)=\theta_s^{(\omega)}(e_\omega ae_\omega)$, and that $\theta_s^{(\nu )}(e_\omega \lambda_\nu(t))=\theta_s^{(\omega )}(\lambda_\omega(t))$: Since $e_\omega\in\qM$ we trivially have that $\theta_s^{(\nu )}(e_\omega)=e_\omega$. This then in turn ensures that $\theta_s^{(\nu)}(e_\omega ae_\omega)=e_\omega ae_\omega=\theta_s^{(\omega)}(e_\omega ae_\omega)$. The second claim follows on noticing that $\theta_s^{(\nu )}(e_\omega \lambda_\nu(t))=e_\omega\theta_s^{(\nu )}(\lambda_\nu(t))=e^{-ist}e_\omega\lambda_\nu(t)=e^{-ist}\lambda_\omega(t)=\theta_s^{(\omega)}(\lambda_\omega(t))$.
\end{itemize}
The above observations suffice to ensure that the subalgebra of $\qM \rtimes_\nu\mathbb{R}$ generated by $e_\omega\qM e_\omega$ and the compressed left-translation operators $e_\omega\lambda(t)e_\omega$ is as a ``crossed product'', an exact copy of $e_\omega\qM e_\omega\rtimes_\omega\mathbb{R}$, which then establishes the theorem.
\end{proof}

\subsubsection{$L^{\cosh-1}$ regularity and subtheories}
We note that the $L^{\cosh-1}$ regularity condition is inherited by certain ``translates'' of open bounded regions $\cO$ for which this condition is known to hold. Specifically given $g= (a,\Lambda) \in P^{\uparrow}_+$ and $\cO\subset \mathbb{M}$, it follows from (L2) that there is an automorphism $\alpha_g$ in $Aut(\qM)$ such that $\alpha_g(\qM(\cO)) = \qM(g\cO)$. If $g$ is chosen so that $\alpha_g$ and $\sigma_t$ commute, then we may once again use \cite[Theorem 4.1]{HJX} and \cite[Corollary 4.5]{LM2} to show that the action of the automorphism $\alpha_g$ extends to a map $\widehat{\alpha}_g$ which maps $\mathfrak{M}(\cO) \rtimes_{\sigma} \Rn$ onto $\mathfrak{M}(g\cO) \rtimes_{\sigma} \Rn$, and $h_{\cO}= \frac{d\hat{\omega}_{\cO}}{d\tau_{\cO}}$ onto $h_{g\cO}=\frac{d\hat{\omega}_{g\cO}}{d\tau_{g\cO}}$. So $\phi(f)$ is a field operator affiliated to $\mathfrak{M}(\cO)$, an $L^{\cosh-1}$ regularity restriction on $\phi(f)$ in terms of $\qM(\cO)$, translates to an $L^{\cosh-1}$ regularity restriction on $\widehat{\alpha}_g(\phi(f))$ in terms of $\qM(g\cO)$.

The concept of $L^{\cosh-1}$-regularity may of course be considered in terms of any of the local algebras $\qM(\cO)$. However in some sense local algebras $\qM(W)$ corresponding to a wedge in Minkowski space occupy a special place in this regard. (See for example the result of Bisognano and Wichman (Theorem \ref{BW}).) As we shall see below, such algebras serve as the context for proving that the concept of $L^{\cosh-1}$ regularity, behaves well with regard to ``subtheories''.

Before we are able to prove such good behaviour of $L^{\cosh-1}$-regularity with respect to subtheories, some background is necessary. In particular we need Borchers' concept of ``modular covariant subalgebras'', and a description of how that concept relates to the construction of subtheories. (See sections VI.1 and VI.3 of \cite{Bor2}.)

In the language of Borchers, a modular covariant von Neumann subalgebra $\qN$ of a von Neumann algebra $\qM$, is an algebra which satisfies the requirement that $\Delta_{\qM}^{it}\qN\Delta^{-it}_{\qM}=\qN$ for all $t\in \mathbb{R}$ \cite[Def VI.1.1]{Bor2}. But as Borchers notes 
\cite[Theorem VI.1.3]{Bor2}, it follows from a result of Takesaki \cite[Theorem IX.4.2]{Tak}, that this requirement is equivalent to the existence of a faithful normal conditional expectation $\mathcal{E}$ from $\qM$ onto $\qN$, which leaves the canonical faithful normal weight $\omega$ on $\qM$, invariant in the sense that $\omega =\omega \circ \mathcal{E}$.

To apply these concepts to local algebras, one needs the notion of ``coherent'' subalgebras. Specifically, following Borchers \cite[Def VI.3.1]{Bor2}, we make the following definition:

\begin{definition}
Let $\qM(W)$ be a von Neumann algebra associated with a wedge $W$, and admitting a cyclic and separating vector. We say that a collection of modular covariant subalgebras $\qN(W)$ of $\qM(W)$ (where $W$ ranges over all wedges in Minkowski space), is coherent if all of the projections $E_W$ from $\mathcal{H}$ onto $[\qN(W)\Omega]$ coincide. 
\end{definition}

In \cite{Bor2} Borchers defines the concept of coherence for subalgebras corresponding to double cones as well. However for the purpose of recovering his main theorem regarding subtheories (stated below), we only need information regarding the coherence of subalgebras corresponding to wedges. Following Borchers, we assume that the considered quantum field fulfills the Bisognano-Wichmann property, i.e. that \emph{for every wedge the modular groups acts locally like the associated group of Lorentz boosts on the underlying space} \cite[Definition III.1.4]{Bor2}. Furthermore again following Borchers, we set $\qM(D)=\cap\{\qM(gW): D\subset gW, g\in P^{\uparrow}_+\}$.

\begin{theorem}[{\cite[Theorem VI.3.5]{Bor2}}]\label{B-subtheory} 
Let $\{\qM(D), U(\Lambda,x), \mathcal{H}, \Omega\}$ be a family of local observables fulfilling the conditions (L1) -- (L6) described in 
the introduction, and assume that this family is in the vacuum sector. (In other words that each $\qM(D)$ is a von Neumann algebra and 
that $\Omega$ is cyclic and separating.) Given any coherent family $\qN(W)$ of modular covariant subalgebras of $\qM(W)$ (where $W$ ranges over all wedges 
in Minkowski space), there exists a family of local observables $\{\qN(D), U(\Lambda,x), E\mathcal{H}, \Omega\}$ fulfilling those same 
conditions. In other words $\{\qN(D), U(\Lambda,x), E\mathcal{H}, \Omega\}$ determines a subtheory. In particular for every wedge we have that 
$\qN(W)=\bigvee\{\qN(D): D\subset W\}$. \newline (Here $\qN(D)$ is defined to be $\cap\{\qN(W): D\subset W\}$, and $E$ is the projector onto $[\qN(W)\Omega]$.)
\end{theorem}

We may now further refine the above theorem to show that such subtheories also behave well with regard to the $L^{\cosh-1}$-regularity condition. 

\begin{theorem} 
Assume that the conditions stated in Theorem \ref{2.3} hold and that $\{\qM(D), U(\Lambda,x), \mathcal{H}, \Omega\}$ is the net of local algebras implied by that theorem. Let $\{\qN(D), U(\Lambda,x), E\mathcal{H}, \Omega\}$ be a subtheory constructed as described in the preceding theorem from some coherent family $\qN(W)$ of modular covariant subalgebras of $\qM(W)$ (where $W$ ranges over all wedges in Minkowski space)

Let $\phi(f)$ denote the extension of the field operators described in Theorem \ref{2.3} which are affiliated to $\qM(W)$.

Then for any wedge $W$ or double cone $D$, a field operator $\phi(f)$ which is affiliated to $\qN(W)$ rather than $\qM(W)$ (respectively $\qN(D)$) will fulfill an $L^{\cosh-1}$-regularity condition with respect to $\qN(W)$ (respectively $\qN(D)$) if and only if it fulfills such a condition with respect to $\qM(W)$ (respectively $\qM(D)$).
\end{theorem}

\begin{proof} The proofs being entirely analogous, we prove the claim for some double cone $D$. To start off with we note that the hypothesis of the theorem ensures that there exists a faithful normal conditional expectation $\mathcal{E}_D$ from $\qM(D)$ onto 
$\qN(D)$, which leaves the faithful normal state $\omega$ on $\qM(D)$ invariant. (See \cite[Lemma 6.3.3]{Bor2}.)

It now follows from (\cite[Theorem 4.1]{HJX}, \cite[Corollary 5.5]{LM2} \cite[Lemma 4.8]{Gol}), that we may regard $\cN_{D}=\qN(D) \rtimes \Rn$ as a subspace of $\cM_D=\qM(D) \rtimes \Rn$ in such a way that the respective canonical traces satisfy $\tau_{\cN}=\tau_{\cM}|_{\cN_{D}}$, with in addition $h_{\cN_D} = \frac{d\hat{\omega}_{\qN}}{d\tau_{\cN}} = \frac{d\hat{\omega}_{\qM}}{d\tau_{\cM}} = h_{\cM_D}$. Once these identifications have been made, the conclusion of the theorem is obvious.
\end{proof}

\subsection{Consequences of $L^{\cosh-1}$ regularity of field operators}

There is further evidence that this restriction is a physically reasonable one to make for field operators, namely that membership of  
$L^{\cosh-1}(\mathfrak{M}(D))$, will ensure that the ``generalised moments'' of the field operators are all finite. This restriction may therefore be seen as a requirement which complements the requirements noted in for example Equation \ref{fmoment} and Axiom (F1). What we mean by this statement, is contained in the following Proposition. (For the sake of simplicity of notation, we will simply write $\mathfrak{M}$ for $\mathfrak{M}(D)$ in the remainder of this subsection.

\begin{proposition} Let $\mathfrak{M}$ be $\sigma$-finite and let $\cM$ and $h$ be as before. If for each $f$ we have that $\varphi_{\cosh-1}(h)^{1/2}\phi(f)\varphi_{\cosh-1}(h)^{1/2}\in L^{\cosh-1}(\mathfrak{M})$, then for any $1\leq p<\infty$ we will have that $h^{1/(2p)}\phi(f)h^{1/(2p)}\in L^p(\mathfrak{M})$. Hence in their action on $\mathfrak{M}$, the generalised moments of the field operators will all be finite, in the sense that for any $a\in \mathfrak{M}$, $tr(a[h^{1/(2p)}\phi(f)h^{1/(2p)}]^p)$ will always be a finite complex number. (Here $tr$ is the tracial functional on $L^1$, not $\tau$.) 
\end{proposition}

\begin{proof}
Given any $2\leq p <\infty$, select $m\in \mathbb{N}$ so that $2m\leq p <2(m+1)$. Then of course $t^p\leq t^{2m}+t^{2(m+1)}$ for all $t\geq 0$. On considering the Maclaurin expansion of $\cosh(t)-1$, it now trivially follows that $\frac{1}{(2(m+1))!}t^p \leq\frac{1}{(2(m+1))!}(t^{2m}+t^{2(m+1)})\leq \cosh(t)-1$ for all $t\geq 0$. It is a straightforward exercise to conclude from this 
fact that $\left(\frac{t}{(2(m+1))!}\right)^{1/p}\leq [\mathrm{arccosh}(t^{-1}+1)]^{-1}=\varphi_{\cosh-1}(t)$. In other words the function
$\gamma_p(t)= \frac{t^{1/p}}{\varphi_{\cosh-1}(t)}$ is a well-defined bounded continuous function on $(0,\infty)$. Hence $\gamma_p(h)$ is a bounded operator. So if $\varphi_{\cosh-1}(h)^{1/2}\phi(f)\varphi_{\cosh-1}(h)^{1/2}$ is $\tau$-measurable, then so is $\gamma_p(h)^{1/2}[\varphi_{\cosh-1}(h)^{1/2}\phi(f)\varphi_{\cosh-1}(h)^{1/2}]\gamma_p(h)^{1/2}= h^{1/(2p)}\phi(f)h^{1/(2p)}$. It is now a simple matter to verify that this $\tau$-measurable operator satisfies all the membership criteria for $L^p(\mathfrak{M})$. 

It remains to consider the case $1\leq p <2$. For $\sigma$-finite algebras it is however known that $L^2(\mathfrak{M})$ contractively embeds into any $L^p(\mathfrak{M})$ (where $1\leq p <2$) in a way which will send $h^{1/4}\phi(f)h^{1/4}$ to $h^{1/(2p)}\phi(f)h^{1/(2p)}$ (see \cite{kos}). This proves the proposition.
\end{proof}

\subsection{A definition of entropy for type III algebras}

Although thus far our focus has been on analysing the link between field operators and the space $L^{\cosh-1}(\mathfrak{M})$, in concluding this section we make some brief comments about the space $L\log(L+1)(\mathfrak{M})$. This is an Orlicz space produced by the Young's function $t\log(t+1)$. To clarify our interest in this Orlicz space, we remind the reader that entropy is a crucial tool in the thermodynamical description of quantum systems. However the standard approach to statistical mechanics leads to serious problems with the definition of entropy (see Wehrl \cite{weh}). The Orlicz space $L\log(L+1)$, being dual to $L^{\cosh-1}$, can be considered to be the natural home for the states acting on regular observables. Importantly in \cite{ML} a strong case was made that the space $L^1\cap L\log(L+1)$ is home for the states with good entropy (see \cite[Proposition 6.8]{ML}).  So in addition to the duality noted above, the space $L\log(L+1)$ can also be considered as the space generated by states with ``good'' entropy. For the sake of presenting a complete physical framework, we therefore pause to indicate how the technology of Orlicz spaces for type III algebras may be used to identify a class of regular states for which continuous entropy is well-defined in even this very general context. This issue will be investigated in more detail in a follow-up paper \cite{ML2}, wherein both relative entropy and continuous entropy for general quantum systems will be considered.  

Let $\Psi_1$ and $\Psi_2$ be two Young's functions. Then $\Psi_1\vee\Psi_2$ defined by $(\Psi_1\vee\Psi_2)(t)=\max(\Psi_1(t),\Psi_2(t))$, 
can be show to also be a Young's function. One may now argue from Definition \ref{orliczelement}, to see that in the classical context a measurable function $f$ belongs to $L^{\Psi_1}\cap L^{\Psi_2}$, if and only if it belongs to $L^{\Psi_1\vee\Psi_2}$. Thus one may regard the intersection $L^{\Psi_1}\cap L^{\Psi_2}$, as another Orlicz space produced by the Young's function $\Psi_1\vee\Psi_2$. On applying this to the intersection $L^1\cap L\log(L+1)$, it follows that classically this intersection may be regarded as an Orlicz space produced by the Young's function
$$\Psi_{ent}(t)=\max(t,t\log(t+1))=\left\{ \begin{array}{ll} t & 0\leq t\leq e-1\\ t\log(t+1) & e-1\leq t\end{array} \right. $$
When regarded as an Orlicz space, we shall write $L^{ent}$ for $L^1\cap L\log(L+1)$. In the setting of general von Neumann algebras, it is the noncommutative analogue of $L^{ent}$ that will form the home for ``states'' with good entropy.
Getting back to the classical setting, we will for simplicity of computation assume that each of $L\log(L+1)(0,\infty)$ and $L^1\cap L\log(L+1)(0,\infty)$ are equipped with the 
Luxemburg norm. It is then an exercise to see that the fundamental function $\varphi_{ent}(t)=\frac{1}{\Psi_{ent}^{-1}(1/t)}$ of $L^{ent}(0,\infty)$ is of the form $\varphi_{ent}(t)=\max(t,\varphi_{\log}(t))$. It is this fundamental function that we use to construct our type III analogue of $L^{ent}$ by means of the prescription given earlier. Let us denote this space by $L^{ent}(\mathfrak{M})$.

From the above computations, it is clear 
that the functions $\zeta_1(t)= \frac{t}{\varphi_{ent}(t)}$, and $\zeta_{\log}(t)= \frac{\varphi_{\log}(t)}{\varphi_{ent}(t)}$ are both continuous and bounded above (by 1) on $(0,\infty)$. Hence the operators $\zeta_1(h)$ and $\zeta_{\log}(h)$ are both contractive 
elements of $\cM$. It is now an exercise to see that the prescriptions $\iota_1:x\to\zeta_1(h)^{1/2}x\zeta_1(h)^{1/2}$ and $\iota_{\log}:x\to\zeta_{\log}(h)^{1/2}x\zeta_{\log}(h)^{1/2}$ respectively yield continuous embeddings of $L^{ent}(\mathfrak{M})$ into 
$L^1(\mathfrak{M})$ and $L\log(L+1)(\mathfrak{M})$. (The embeddings are clearly bounded. One just needs to check that the images satisfy the membership criteria for $L^1(\mathfrak{M})$ and $L\log(L+1)(\mathfrak{M})$.) Let $x\in L^{ent}(\mathfrak{M})^+$ be given. For such an element, the quantity
$$\inf_{\varepsilon >0}[\varepsilon\tau(E^{|\iota_{\log}(x)|}(\varepsilon, \infty)) + \log(\varepsilon)\|\iota_1(x)\|_1]$$will then be well-defined (albeit possibly infinite-valued). To see this observe that the fact that $\iota_{\log}(x)\in L\log(L+1)(\mathfrak{M})$,  ensures that for some $\epsilon>0$, $\tau(E^{|\iota_{\log}(x)|}(\varepsilon, \infty))$ is finite, with $\|\iota_1(x)\|_1$ clearly always finite since $\iota_1(x)\in L^1(\mathfrak{M})$.) \textbf{It is this quantity that we propose as the type III analogue of von Neumann entropy.} 

To see that this does indeed make sense, recall that in the case where $\mathfrak{M}$ is semifinite and $\omega$ a trace, $\cM$ may be identified with $\qM\otimes L^\infty(\mathbb{R})$, and that with respect to this identification, any given Orlicz space $L^\Psi(\mathfrak{M})$ will consist of the simple tensors in $\qM\otimes L^\infty(\mathbb{R})$ of the form 
$f\otimes \varphi_{\Psi}(e^t)$, where $f$ is an element of the ``tracial'' Orlicz space $L^\Psi(\mathfrak{M},\omega)=\{g\in\widetilde{\mathfrak{M}}: \Psi(\lambda|g|)\in \widetilde{\mathfrak{M}}\mbox{ and }\omega(\Psi(\lambda|g|))<\infty\}$. In this setting we also have that $h=\I\otimes e^t$. So given some $x=f\otimes \varphi_{ent}(e^t)$ in $L^{\Psi_{ent}}(\mathfrak{M})$ (equivalently $f\in L^{\Psi_{ent}}(\mathfrak{M},\omega)$), we may use these facts to show that here $\iota_1(f\otimes \varphi_{\Psi}(e^t))=f\otimes e^t$ and $\iota_{\log}(f\otimes \varphi_{\Psi}(e^t))=f\otimes \varphi_{\log}(e^t)$. For such an $x=f\otimes \varphi_{ent}(e^t)$ in $L^{\Psi_{ent}}(\mathfrak{M})$, it follows from \cite[Corollary 2.3]{L} that the quantity \begin{eqnarray*}
&& \inf_{\varepsilon >0}[\varepsilon\tau(E^{|\iota_{\log}(x)|}(\varepsilon, \infty)) + \log(\varepsilon)\|\iota_1(x)\|_1]\\
&& = \inf_{\varepsilon >0}[\varepsilon\tau(E^{|f\otimes \varphi_{\log}(e^t)|}(\varepsilon, \infty)) + \log(\varepsilon)\|f\otimes e^t\|_1]\\
\end{eqnarray*}
may be rewritten as 
$$\inf_{\varepsilon >0}[\varepsilon\omega((f/\varepsilon)\log((f/\epsilon)+\I)) +\log(\varepsilon)\omega(f)]=\inf_{\varepsilon >0}\omega(f\log(f+\varepsilon\I)).$$By \cite[Proposition 6.8]{ML}, this yields exactly $\omega(f\log(f))$. Hence the above prescription for entropy, is a faithful extension of the existing prescriptions for entropy on semifinite von Neumann algebras. In the general setting, the subspace of $L^1(\mathfrak{M})$ corresponding to ``good'' states which have a well-defined entropy, is therefore given by $\iota_1(L^{\Psi_{ent}}(\mathfrak{M}))$.

\subsection{Concluding remarks}\label{concludesection2}

The analysis in the preceding sections, leads to the following conclusions:

We begin with a close examination of the significance of the compression used in Theorem \ref{Lcosh-criteria}. In that theorem, we refered to the Osterwalder-Schrader results. These were obtained within the framework of Euclidean field theory. We remind the reader that in the Schwinger approach, this theory is based on the Wightman axioms and (classical) probability theory - see \cite{Nel}, \cite{Simon}.  Here, among other things, we replaced classical probability theory by quantum probability formulated in terms of noncommutative integration theory.

On the other hand, Haagerup's  approach to noncommutative integration is very well suited to type III factors. Furthermore, in this approach the modular structure (as described by Tomita-Takesaki theory) plays a crucial role. Note that in the context of Theorem \ref{Lcosh-criteria}, $\Omega$ is a common cyclic and separating vector for each of the local algebras $\qM(\cO)$ in the net (\ref{net}).
But this implies that for 
$\cO \neq \cO^{\prime}$ one gets two pairs $(\qM(\cO), \Omega)$ and $(\qM(\cO^{\prime}), \Omega)$ with, in general,  $\qM(\cO) \neq \qM(\cO^{\prime})$. Thus the modular structures associated with $(\qM(\cO), \Omega)$ and $\qM(\cO^{\prime}), \Omega)$ are different - see the warning given prior the proof of Theorem 6.3.28 in \cite{BR}.

It is therefore important point to observe that here the employment of the compression given in Theorem \ref{Lcosh-criteria} radically 
changes the picture. Specifically here the subalgebra of $B(e_{\omega} \cH)$ generated by $e_{\omega} \qM(\cO) e_{\omega}$ is strongly 
dense in $e_{\omega} \qM e_{\omega}$. To see this it is enough to carefully apply Theorem 9.2.36 of \cite{KRi}.

Consequently, in the compressed structures one gets a unique modular structure, and hence a unique basis for noncommutative integration. This explains why the employment of the compression in Theorem \ref{Lcosh-criteria} was both necessary and expected.

The preceding analysis also leads to the following further conclusions:

\begin{corollary}
\label{1}
The $\tau$-measurability regularity conditions on the field operators are strongly related to the principle of relativity, i.e. on the one hand the affiliation of the field operators to the local algebra $\mathfrak{M}(D)$ results from some form of regularity with respect to the generator $H$ of time translations, whilst the embedding of the field operators into the space $L^{\cosh-1}(\mathfrak{M}(D))$, is a consequence of regularity with respect to the generator $K$ of boosts. Consequently, both of the basic transformations employed by the principle of relativity are used!
\end{corollary}

and

\begin{corollary}
\label{2}
The axioms of QFT imply that field operators should have all moments finite. This requirement is intimately related to the requirement that these operators embed into the space $L^{\cosh-1}(\mathfrak{M})$. The states which come from the space $L^{ent}(\mathfrak{M})$ (the type III analogue of $L^1\cap L\log(L+1)$ described above) moreover all admit a good definition of entropy. Consequently, the strategy based on the quantum Pistone-Sempi theorem, see \cite{LM}, \cite{ML} leads to the conjecture that the proper formalism for QFT is that based on the pair of quantum Orlicz spaces $\langle L^{\cosh-1}(\mathfrak{M}), L\log(L+1)(\mathfrak{M})\rangle$.
\end{corollary}

\section{Tangentially conditioned algebras: an interlude on comparing $\qM(\mathbb{M})$ and $\qM(M)$}

Throughout this section and the next, $M$ will be a (smooth) $d$-dimensional, connected time-oriented globally hyperbolic Lorentzian manifold, $\mathbf{g}$ a Lorentzian metric on $M$, and $\mathbb{M}$ d-dimensional Minkowski space-time. 

The axiomatisation of local algebras for manifolds rather than Minkowski space, was pioneered by Dimock, et al, as early as 1980 \cite{D1, D2}. However it was not until 2003 that Brunetti, Fredenhagen and Verch added the all important principle of locality to the covariance axioms for these algebras. For the sake of background we review some material from \cite{BFV}, \cite{BF}. All these results are formulated in strongly categorical language. However the crucial result for us, is the description of such local algebras given in \cite[Theorem 2.3]{BFV}. For a local algebra fulfilling their criteria, Brunetti, Fredenhagen and Verch obtain the result below. This encapsulates the structural information which will form the starting point of our subsequent modelling. Since in their theory the local algebras can be either $C^*$-algebras or von Neumann algebras, we here follow their convention of denoting the local algebras by $\cA(\cO)$. Here $\cK(M)$ denotes the set of all open subsets in $M$ which are relatively compact and which for each pair of points $x$ and $y$, also contain all $\mathbf{g}$-causal 
curves in $M$ connecting $x$ and $y$. Moreover, as before $\cA(\cO)$ denotes the $C^*$-algebra generated by field operators $\phi(f)$ with test function supported in $\cO$, i.e. $\mathrm{supp} f \subseteq \cO$. We emphasize that $M$ is now a much more general manifold than $\mathbb{M}$, which was the focus in the previous sections.

\begin{theorem}\cite{BFV} \label{localnet}
Let ${\mathscr A}$ be a covariant functor with the properties stated in \cite[Def. 2.1]{BFV}, 
and define a map $\cK(M,\mathbf{g}) \owns O \mapsto \cA(O) \subset
{\mathscr A}(M,\mathbf{g})$ by setting
$$ \cA(O) := \alpha_{M,O}({\mathscr A}(O,\mathbf{g}_O))\,,$$
having abbreviated $\alpha_{M,O} \equiv \alpha_{\iota_{M,O}}$.
The following statements hold:
\begin{itemize}
\item[(a)] Isotony, i.e.
 $$\cO_1 \subset \cO_2 \Rightarrow
\cA(\cO_1) \subset \cA(\cO_2) \qquad \mbox{for all}\qquad  \cO_1,\cO_2 \in
\cK(M)\,.$$ 
\item[(b)] 
If there exists a group $G$ of isometric diffeomorphisms
$\kappa: M \to M$ (so that $\kappa_*{\mathbf{g}} = {\mathbf{g}}$) preserving orientation
and time-orientation, then there is a representation by C$^*$-algebra automorphisms
${\alpha}_{\kappa} : \cA(M) \to \cA(M)$ 
such that
$$
{\alpha}_{\kappa}(\cA(\cO)) = \cA(\kappa(\cO))\,, \quad \cO \in
\cK(M)\,.
$$
\item[(c)] If the algebra belongs to the ``causal'' category of their theory, then it also holds that 
$$ [\cA(\cO_1),\cA(\cO_2)] = \{0\} $$
for all $\cO_1,\cO_2 \in \cK(M)$ with $\cO_1$ causally separated from
$\cO_2$.
\item[(d)] If the theory of which this algebra is part fulfills the time-slice
  axiom, and 
$\Sigma$ is a Cauchy-surface in $M$ with $S
  \subset \Sigma$ open and connected, then for each $\cO \in
  \cK(M)$ with $\cO \supset S$ it holds that
 $$ \cA(\cO) \supset \cA(S^{\perp}{}^{\perp}) $$
where $S^{\perp}{}^\perp$ is the double causal complement of $S$, and
$\cA(S^{\perp}{}^{\perp})$ is defined as the smallest $C^*$-algebra
formed by all $\cA(\cO_1)$, $\cO_1 \subset S^{\perp}{}^{\perp}$, $\cO_1
\in \cK(M)$.
\end{itemize}
\end{theorem} 

Appealing as this paradigm may be, no manifold based theory is complete without a clear concept of vector bundles. With this in mind we wish to investigate the extent to which tangential phenomena may be encoded at the algebra level. This consideration will ultimately lead us to the identification of a large class of algebras $\cA(M)$ which locally ``look like'' algebras of the form $\cA(\mathbb{M})$, and for which the conclusions of the preceding section are therefore applicable. To achieve this objective, we go back to the example that inspired both Dimock, and Brunetti, Fredenhagen and Verch to formulate their approaches to local algebras on Lorentzian manifolds. To formulate the next result some preliminaries are necessary. 
We wish to consider local algebras generated by field operators which are solutions of the Klein-Gordon equation.
As was stated, $M$ stands for a manifold satisfying the conditions given at the beginning of this section. $C^{\infty}_0(M)$ will denote the space of smooth, real valued functions on $M$ which have compact support. 
Following Dimock \cite{D1} (see also \cite{BFV}), we will describe the CCR algebra of bosonic fields on the manifold $M$ given by solutions of the Klein-Gordon equation. 
The crucial point in his construction, is that global hyperbolicity of $M$ entails existence of global fundamental solutions $E$ for the Klein-Gordon equation $(\Box + m^2 + \xi R)\phi = 0$, where $m\geq 0$ , $\xi \geq0$ are constants, and $R$ is the scalar curvature of the metric on $M$. In particular, $E = E^{adv} - E^{ret}$, where $E^{adv/ret}$ (advanced/retarded, 
respectively) are well defined maps such that $E^{adv/ret}: C^{\infty}_0(M) \to C^{\infty}(M)$. 

It was shown by Dimock \cite{D1}, that the property of bosonic field operators being solutions of the Klein-Gordon equation, is characterized by the following relations:
\begin{equation}
e^{i\phi(f)} e^{i \phi(f^{\prime})} = e^{i\phi(f^{\prime})} e^{i \phi(f)} e^{-i\langle f,Ef^{\prime}\rangle}.
\end{equation}

This leads to the following form of Weyl relations for $W(f) \equiv e^{i\phi(f)}$:

\begin{equation}
W(f) W(f^{\prime}) = e^{- \frac{i}{2} \langle f, E f^{\prime}\rangle} W(f+f^{\prime}).
\end{equation}
Denote by $\cR$ the real vector space $E\left(( C^{\infty}_0(M)\right)$. We note that $\sigma(f,g) = \int_M f Eg dV_M$, where $dV_M$ is the volume form on $M$, gives a symplectic form on $\cR$, i.e. $(\cR, \sigma)$ is the symplectic space.
The $C^*$-algebra of canonical commutation relations over a symplectic form $(\cR, \sigma)$, writen as $\mathfrak{M}(\cR, \sigma)$, is by definition the $C^*$-algebra generated by elements
\begin{equation}
W(-f) = W(f)^*, \quad W(f) W(g) = e^{-\frac{i}{2} \sigma(f,g)} W(f+g).
\end{equation}

Let $(M_1, g_1)$  and $(M_2, g_2)$ be two given manifolds satisfying the prescribed conditions. Assume that $\psi: M_1 \to M_2$ is a diffeomorphism satisfying $\psi_*g_1 = g_2|_{\psi(M_1)}$ and which also preserves causality and orientation; see \cite{BFV} for details. Denote by $E^{\psi}$ the global fundamental solution of the Klein-Gordon equation on $\psi(M)$, and write $\cR^{\psi} \equiv E^{\psi}\left(C^{\infty}_0(\psi(M))\right)$. Dimock \cite{D1}, has shown that $E^{\psi} = \psi_* \circ E \circ \psi^{-1}_*$ and $\cR^{\psi} = \psi_* \cR$, where $\psi_* f = f \circ \psi^{-1}$. Moreover
\begin{equation}
\sigma(Ef,Eg) = \sigma^{\psi}(\psi_*Ef,\psi_*Eg),
\end{equation}
where $\sigma^{\psi}$ is a symplectic form on $\cR^{\psi}$ which is defined in an analogous fashion to $\sigma$. Then the prescription
\begin{equation}\label{3.5}
\tilde{\alpha}_{\psi}(W(f)) = W^{\psi}(\psi_*f), \quad f \in \cR,
\end{equation}
for the generators $\{W^{\psi} \}$ of the CCR algebra over $(\cR^{\psi}, \sigma^{\psi})$, leads to $^*$-isomorphisms between the corresponding algebras. To summarize, the CCR algebras described above, yield a version of Theorem \ref{localnet} adapted to solutions of the Klein-Gordon equation. The theorem as stated below is due to Brunetti, Fredenhagen and Verch (see \cite[Theorem 2.2]{BFV}). To be faithful to their comprehension and formulation of this result, we first define the category $\textbf{\textsf{Loc}}$

\begin{description}
\item[$\textbf{\textsf{Loc}}$] The class of objects
$\mathrm{obj}(\textbf{\textsf{Loc}})$ is formed by all (smooth) $d$-dimensional 
($d\ge 2$ is held fixed), globally
hyperbolic Lorentzian spacetimes $M$ which are oriented and time-oriented.
Given any two such objects $M_1$ and $M_2$, the morphisms
$\psi\in\mbox{\rm hom}_{\textbf{\textsf{Loc}}}(M_1,M_2)$ 
are taken to be the isometric embeddings
$\psi: M_1 \to M_2$ of $M_1$ into
$M_2$ but with the following constraints;  
\begin{itemize}
\item[$(i)$] if $\gamma : [a,b]\to M_2$ is any 
causal curve and
$\gamma(a),\gamma(b)\in\psi(M_1)$ then the whole curve must be 
in the image $\psi(M_1)$, i.e., $\gamma(t)\in\psi(M_1)$ for all $t\in ]a,b[$; 
\item[$(ii)$] any morphism preserves orientation and
  time-orientation of the embedded spacetime.
\end{itemize} 
Composition is composition of maps, the unit element in $\mbox{\rm hom}_{\textbf{\textsf{Loc}}}(M,M)$ is given by 
the identical embedding ${\rm id}_M : M\mapsto M$ for
any $M\in\mathrm{obj}(\textbf{\textsf{Loc}})$.
\end{description}

\begin{theorem}\label{ccr}
For each $M \in \mathrm{obj}(\textbf{\textsf{Loc}})$ define the $C^*$-algebra
$\cA(M)$ as the CCR-algebra $\mathfrak{W}({\mathcal
  R}(M),\sigma_{M})$ of the Klein-Gordon equation
\begin{equation}\label{KGeqn}
(\nabla^a\nabla_a + m^2 + \xi R)\varphi = 0
\end{equation}
(with $m,\xi$ fixed for all $M$ and $\mathcal{R}$ the scalar curvature), and for each
$\psi \in \mbox{\rm hom}_{\textbf{\textsf{Loc}}}(M,M')$ define the $C^*$-algebraic
endomorphism by $\alpha_{\psi} = \tilde{\alpha}_{\iota_{\psi}} \circ
\tilde{\alpha}_{\psi}: \cA(M) \to \cA(M')$ 
where $\tilde{\alpha}_{\psi}$ and $\tilde{\alpha}_{\iota_{\psi}}$ are respectively given by 
\begin{equation}
\label{Bog1}
 \tilde{\alpha}_{\psi}(W(\varphi)) = W^{\psi}(\psi_*(\varphi))\,, \quad
 \varphi \in {\mathcal R}\,
\end{equation}
and
\begin{equation} \label{Bog2}
\tilde{\alpha}_{\iota_{\psi}}(W^{\psi}(\phi)) =
W'(T^{\psi}\phi)\,,  \quad \phi \in {\mathcal R}^{\psi}\, .
\end{equation}
(Here $W^{\psi}(\,.\,)$ are as before, $W(.)$ the generators of the Weyl-generators of 
$\mathfrak{W}({\mathcal R},\sigma)$, 
and $W'(\,.\,)$ the Weyl-generators of
$\mathfrak{W}({\mathcal R}',\sigma')$, while $T^\psi$ is the corresponding symplectic map from $({\mathcal R}^\psi,\sigma^\psi)$ into $({\mathcal R}',\sigma')$.) 

\textbf{Then the corresponding local algebra fulfills all the criteria of the preceding theorem including causality and the time-slice axiom.}
\end{theorem} 

We pause to point out that we do not in any way wish to add to or analyse the mathematical physical modelling inherent in the above result. Our only objective in stating the above result, is to emphasise the importance of the class of CCR-algebras $\mathfrak{W}({\mathcal R}(M),\sigma_{M})$. Having noted their significance, we show that these algebras are also \emph{tangentially conditioned}, where the (mathematical) concept of tangential conditioning of local algebras is defined as below.  

\begin{definition}\label{tancon}
We say that a local algebra $\cA(M)$ is \emph{tangentially conditioned} if $\cA(M)$ behaves well with respect to the atlas on $M$ in the following sense: Given any point $p\in M$, there exists a triple $(\iota_p, \cO_p, \tilde{\cO}_0)$ from this atlas, where $\cO_p$ is a neighbourhood of the point $p\in M$, $\tilde{\cO}_0$ a corresponding diffeomorphic neighbourhood of $0\in \mathbb{M}$, and $\iota_p$ the diffeomorphism on $\tilde{\cO}_0$ identifying these neighbourhoods, such that the algebras $\cA(\cO_p)$ and $\cA(\tilde{\cO}_0)$ are $*$-isomorphic by means of a $*$-isomorphism $\beta_p$ implemented by the diffeomorphism in the sense that for any open subset $\cO_1$ of 
$\cO_p$, the restriction of $\beta_p$ to $\cA(\cO_1)$ yields a $*$-isomorphism from $\cA(\cO_1)$ onto $\cA(\iota_p(\cO_1))$.
\end{definition}

It is clear from the definition that tangentially conditioned algebras $\cA(M)$, are those which at a local level ``look like'' the algebra $\cA(\mathbb{M})$. Thus for this class of algebras, the conclusions of the preceding section are at a local level immediately applicable. The important fact to note, is that the CCR-algebras considered earlier in 
Theorem \ref{localnet} provide concrete examples of tangentially conditioned algebras!
Theorem \ref{ccr} also provides concrete examples of tangentially conditioned algebras, on condition that the Klein-Gordon equation is locally solvable, i.e. there is an atlas such that for each chart $(\cO_p, \iota_p)$ there is a (local) fundamental solutions 
$E_{\cO_p}$. In other words, $\cO_p$ considered as a manifold should enjoy all manifold properties assumed in Theorem \ref{ccr}.

\begin{theorem}
\label{1.4}
The CCR algebras
\begin{enumerate}
\item $\cA(M)=\mathfrak{W}((M),\sigma_{M})$  generated on the symplectic space \newline
$\left( \mathrm{span}_{\bC}\left(C^{\infty}_0(M)\right), \sigma(f,g) = Im \int_M \overline{f} g dV_M \right)$; $\mathrm{span}_{\bC}$ stands for the complex span of $C^{\infty}_0(M)$,
\item $\cA_0(M)=\mathfrak{W}({\mathcal R}(M),\sigma^E_{M})$ generated on the symplectic space \newline 
$\left( C^{\infty}_0(M), \sigma(f,g) = \int_M f Eg dV_M\right)$ for the locally solvable Klein-Gordon equation,
\end{enumerate}
are tangentially conditioned.
\end{theorem}
\begin{proof}
To prove the first claim, let us consider a quadruple $(p, \iota_p, \cO_p, \tilde{\cO}_0)$ where, as in Definition \ref{tancon},  $\cO_p$ is a neighborhood of some $p \in M$, $\tilde{\cO}_0$ a corresponding diffeomorphic neighborhood of $0 \in \mathbb{M}$, and $\iota_p :\cO_p \to \tilde{\cO}_0$ the diffeomorphism on $\cO_0$ 
identifying these neighborhoods. The subspace formed by $\{f\in C^{\infty}_0(M); \mathrm{supp} f \subseteq \cO_p \}$ ($\{f\in C^{\infty}_0(\mathbb{M}); \mathrm{supp} f \subseteq \tilde{\cO}_0 \}$) will be denoted by $\cH_p$ (by $\cH_0$ respectively).

Let $g$ be the Lorentzian metric for $M$. Then the volume element for $M$ is of the form $$dV_M=\sqrt{|\det(g_{ij}(x))|} (\iota_p)_*dV_{\mathbb{M}}(\iota_p^{-1} (x))$$where $(\iota_p)_*$ is the pull-back of $\iota_p$. (See page 433 of \cite{Wal}.) By the definition of the metric tensor, $\det(g_{ij}(x))$ is a smooth function on $\tilde{\cO}_0$. But from the discussion on page 433 of \cite{Wal} it is clear that $\det(g_{ij})$ is non-zero on $\tilde{\cO}_0$. In the case of a Riemannian metric this fact follows from the fact that the matrix $[g_{ij}(x)]$ is positive-definite -- see the discussion following Definition 13.2 on page 403 of Gallier and Quaintance's comprehensive online book \cite{GQ}. So $\frac{1}{\det(g_{ij})}$ is also a smooth function on $\tilde{\cO}_0$. By passing to a smaller pair of neighbourhoods (if necessary) for which the pair of closures are inside the original pair $(\cO_p, \tilde{\cO}_0)$,
we may assume that both $\cO_p$ and $\tilde{\cO}_0$ have compact closures, that $\iota_p$ extends to a homeomorphism identifying these 
two closures, and that both $\det(g_{ij})$ and $\frac{1}{\det(g_{ij})}$ are continuous on the compact closure of $\tilde{\cO}_0$. We may then use the Stone-Weierstrass theorem to select sequences of polynomials $\{p_n\}$ and $\{\tilde{p}_n\}$ such that $\{p_n(\det(g_{ij})\}$ 
uniformly converges to $|\det(g_{ij})|^{1/4}$ on the closure of $\tilde{\cO}_0$, and $\{\tilde{p}_n(1/\det(g_{ij}))\}$ uniformly to 
$|1/\det(g_{ij})|^{1/4}$. 

For every smooth function $f$ with compact support inside $\tilde{\cO}_0$, $p_n(\det(g_{ij}))f$ will be another such function. These functions are of course dense in $\cH_0$. It is now easy to check that for any $n,m$ we have that $\|p_n(\det(g_{ij}))f -p_m(\det(g_{ij}))f\|_2\leq \|p_n(\det(g_{ij})) -p_m(\det(g_{ij}))\|_\infty\|f\|_2$. Thus $\{p_n(\det(g_{ij}))f\}$ is a Cauchy sequence with pointwise limit $|\det(g_{ij})|^{1/4}f$. Hence $|\det(g_{ij})|^{1/4}f$ must be the limit in $\cH_0$ of the sequence $\{p_n(\det(g_{ij}))f\}$. For a general element $f$ of $\cH_0$, select a sequence $\{f_n\}$ converging to $f$ in $\cH_0$. The fact that for any $n,m$ we have that $\|\,|\det(g_{ij})|^{1/4}f_n-|\det(g_{ij})|^{1/4}f_m\|_2\leq \|\,|\det(g_{ij})|^{1/4}\|_\infty\|f_n-f_m\|_2$, ensures that $\{|\det(g_{ij})|^{1/4}f_n\}$ is a Cauchy sequence. Once again one can use classic measure theoretic results to conclude that this sequence must converge pointwise almost everywhere to its limit in $\cH_0$, which can then only be $|\det(g_{ij})|^{1/4}f$. Thus $|\det(g_{ij})|^{1/4}f\in\cH_0$ whenever $f\in \cH_0$. But the same type of argument shows that multiplication by $\frac{1}{|\det(g_{ij})|^{1/4}}$ will map $\cH_0$ back into itself. Thus both of these multiplication maps are bijections.

Writing $w(s)$ for $(|\det(g_{ij}(s))|)^{1/4}$, we now define the map $T:\cH_p \to \cH_0$ by
\begin{equation}
(Tf)(x) = w(x)f(\iota_p^{-1} x), \quad x \in \tilde{\cO}_0.
\end{equation}
It now follows from the preceding discussion that $T$ is a linear bijection from $\cH_p$ onto $\cH_0$. We will show that $T$ is in fact an isometry. To this end we note that by equation (B.2.17) in appendix B of \cite{Wal}, 
\begin{eqnarray*}
\langle Tf,Tg\rangle_{\cH_0} &=& \int_{\tilde{\cO}_0} \overline{(Tf)}(x)(Tg)(x)dV_{\mathbb{M}}(x)\\
&=& \int_{\tilde{\cO}_0}\overline{f}\circ \iota_p^{-1}(x) g\circ\iota_p^{-1}(x)w(x)^2 dV_{\mathbb{M}}(x)\\
&=& \int_{\tilde{\cO}_0}\overline{f}\circ \iota_p^{-1}(x) g\circ\iota_p^{-1}(x)|\det(g_{ij})|^{1/2} dV_{\mathbb{M}}\left(\iota_p \circ \iota_p^{-1} (x)\right)\\
&=& \int_{\tilde{\cO}_0}\overline{f}\circ \iota_p^{-1}(x) g\circ\iota_p^{-1}(x)|\det(g_{ij})|^{1/2} (\iota_p)_*dV_{\mathbb{M}}(\iota_p^{-1} (x))\\
&=& \int_{\cO_p} \overline{f} g dV_{{M}} = \langle f,g\rangle_{\cH_p},
\end{eqnarray*}
where $(\iota_p)_*$ is the pull-back of $\iota_p$.

But, then $T$ preserves the symplectic form $\sigma$. Hence, by \cite[Theorem 5.2.8]{BR}, the prescription
\begin{equation}
\alpha_T(\mathfrak{W}(f)) = \mathfrak{W}(Tf), \quad f \in C^{\infty}_0(M)
\end{equation}
yields a *-isomorphism between from $\cA(\cO_p)$ onto $\cA(\tilde{\cO}_0)$.

The second claim follows by arguments given prior to Theorem \ref{ccr}.
\end{proof}

\section{Local flows and graded algebras of differential forms for local algebras}

Having considered the application of integrable structures to local algebras, we now turn our attention to differential structures. 
Such structures are indispensable tools for the description of time evolution of quantum systems. In other words, integrable structures provide a rather static setting for an analysis of systems, whilst differential structures are employed for an 
examination of time evolution of these systems. Throughout $M$ will be a (smooth) $d$-dimensional, connected time-oriented globally hyperbolic Lorentzian manifold, $\mathbf{g}$ a Lorentzian metric on $M$, and $\mathbb{M}$ $d$-dimensional Minkowski space-time. The key tool we shall exploit in this endeavour, is that tangentially conditioned algebras allow for a local action of the Poincar\'e group.

\subsection{Tangentially conditioned algebras and local flows along contours}

We now turn to the question of dynamics on the algebras $\cA(M)$. When studying dynamics, it is important to identify the appropriate mode of continuity with which to describe such a dynamical flow. We are particularly interested in the appropriate mode of continuity that may be assigned to the translation automorphisms in the representation of the Poincar\'e group. Property 1 of the G\"arding-Wightman postulates for field operators (see \cite[\S IX.8]{RS2}), as well as the behavior of local algebras which fulfil the spectrum condition (see \cite[p. 33]{Sak}), both suggest that it is a physically reasonable assumption to make, that these translation automorphisms are implemented by a strongly continuous unitary group acting on the underlying Hilbert space. That translates to strong operator continuity of the translation automorphisms. This mode of continuity is however more suited to a von Neumann algebraic rather than a $C^*$-algebraic framework. \textbf{Hence in the remainder of the paper we will restrict ourselves to von Neumann local algebras}, for 
which the group of translation automorphisms on $\qM(\mathbb{M})$ is strong operator continuous. The only comment we will make about $C^*$-framework, is that each of the subsequent results will under appropriate restrictions admit of a $C^*$-version.

With the framework in which we will work now clear, that leads us to the question of how one may realise a dynamical flow on the manifold 
$M$, at the algebra level. Part (b) of Theorem \ref{localnet} assures us that in the case where one is fortunate enough to have a globally defined group $G$ of isometric diffeomorphisms on $M$ preserving orientation and time-orientation, the dynamics described by that group canonically lifts to the algebra setting. But what is equally clear from this theorem, is that not all algebras $\qM(M)$ have this property. If however one is content to settle for fairly strong locality as far as dynamics is concerned, the situation improves. Specifically for any tangentially conditioned local algebra $\qM(M)$, the ``local'' dynamics on $M$ does indeed lift to the algebra level. We point out that this behaviour is in line with the classical setting. See the discussion on page 35 of \cite{Thi}. 

In order to be able to deliver on our promise, some background is necessary. For general $d$-dimensional $C^\infty$-manifolds $M$ the $C^\infty$ derivations of $C^\infty(M)$ correspond to local flows on $M$. (To see this combine the comment preceding Theorem 2.2.24 of \cite{Thi} with \cite[Remark 2.3.11(1)]{Thi}.) It is moreover known that on any chart $U$ of $M$, these derivations are up to a diffeomorphism of the form $\sum_{i=1}^df_i\frac{\partial}{\partial x_i}$, where for each $i$ we have that $f_i\in C^\infty(U)$ \cite[2.2.27(8) \& 2..4.3(1)]{Thi}. If therefore in the context of $\qM(\mathbb{M})$ we are able to identify the appropriate analogues of $\frac{\partial}{\partial x_i}$ and $C^\infty(U)$, we will at a formal level be able to give a ``chart-wise'' description of the quantum smooth local flows associated with a tangentially conditioned algebra $\qM(M)$. However the manifolds in view here are Lorentzian. We therefore briefly pause to describe how the above idea may be refined to this context.

By \cite[Chapter 2, Theorem 1]{BaF} any (smooth) $d$-dimensional, connected time-oriented globally hyperbolic Lorentzian manifold $M$, may be written in the form $M=\mathbb{R}\times S$ where each ${t}\times S$ is a Cauchy hypersurface. (Here $\mathbb{R}$ models the time variable.) A careful consideration of part (3) of this result, shows that $S$ is in its own right a $(d-1)$--dimensional Riemannian Manifold. So at a local level $S$ ``looks like'' $\mathbb{R}^{d-1}$. Using this fact, we may select our charts for $M$ in such a way that the local diffeomorphisms which compare $\mathbb{M}$ to $M$, maps points of the form $(t, x_1, x_2,\dots, x_{d-1})$ in say $\tilde{\cO}_0\subset \mathbb{M}$, onto points 
$(r(t),s)$ in ${\cO}_p\subset M$, where $(r(t),s)\in \mathbb{R}\times S$, with $t$ going to $r(t)$ and 
$(x_1, x_2,\dots, x_{d-1})$ to $s$. Suppose that this is the case and let $\cO_p$ be a neighbourhood of some $p\in M$ which is in the above sense diffeomorphic to a neighbourhood $\tilde{\cO}$ of $0\in \mathbb{M}$ by means of some diffeomorphism $\iota_p$ of the above type. 

Based on the above discussion and the assumptions made therein, on the chart $\cO_p$ a continuous local dynamical flow along some contour on $M$ passing through $p=(t_p,s_p)$, may in principle be regarded as the image under $\iota_p$ of a continuous local dynamical flow along a contour flowing through $0\in \mathbb{M}$, where the dynamical 
flow on $\mathbb{M}$ corresponds to a set of points $(t, x_1(t), x_2(t),\dots, x_{d-1}(t))\in \tilde{\cO}_0$ which varies continuously as $t$ varies over the interval $(-\varepsilon, +\varepsilon)$, with \newline $(0, x_1(0), x_2(0),\dots, x_{d-1}(0))=0$, and with $\iota_p(0)=p$. For the sake of simplicity let us write $g_t$ for $(t, x_1(t), x_2(t),\dots, x_{d-1}(t))$ and $\beta_p$ for the *-isomorphism from $\qM(\cO_p)$ to $\qM(\tilde{\cO}_0)$. 

In the context of $\qM(\mathbb{M})$, this local dynamics may then formally be lifted to the algebra level by using the ``translation 
automorphisms'' in the representation of the Poincar\'e group on $\qM(\mathbb{M})$. Specifically if on $\mathbb{M}$ the contour is described by the set of points $g_t$ (indexed by the time variable), we may pass to the set $\alpha_{g_t}$ where each $\alpha_{g_t}$ satisfies $\alpha_{g_t}(\qM(\tilde{\cO}))= \qM(g_t+\tilde{\cO})$ ($\tilde{\cO}$ an 
open subset of $\mathbb{M}$). The natural domain of the ``restriction'' of the action of $\alpha_{g_t}$ to the $\qM(\tilde{\cO}_0)$ context, is then $[\qM(\tilde{\cO}_0)\cap\alpha_{g_t}^{-1}(\qM(\tilde{\cO}_0))]$. For any $g_t$ this natural domain will include all 
subalgebras $\qM(\tilde{\cO}_t)$ of $\qM(\tilde{\cO}_0)$, for which $\tilde{\cO}_t$ is a subset of $\tilde{\cO}_0$ small enough to ensure that  $g_t+\tilde{\cO}_t\subset \tilde{\cO}_0$. As $g_t$ gets ``closer'' to 0, we expect the size of the sets $\tilde{\cO}_t$ we are able to select, to increase. We pause to explain how one can make these ideas exact.

By passing to a subset if necessary, we may assume that the set of $\tilde{\cO}_0$ is an open double cone $K=\tilde{\cO}_0$ centred at 0. It is then an exercise to see that $K=\cup_{n=1}^\infty\frac{n}{n+1}K$. We may further find a decreasing sequence $\varepsilon_n>0$ such that $(\frac{n}{n+1}K +g) \subset K$ whenever $g=(t,x_1, \dots, x_{d-1})$ is an element of $\mathbb{M}$ for which 
$\|g\|_2^2=|t|^2 +\sum_{k=1}^{d-1}|x_k|^2<\varepsilon_n$. As far as our local dynamical flow along the given contour is concerned, for each $\varepsilon_n>0$ we may by assumption find some $\delta_n>0$ such that $\|g(t)\|_2^2=|t|^2 +\sum_{k=1}^{d-1}|x_k(t)|^2<\varepsilon_n$ whenever $|t|< \delta_n$. At the operator level, each $\alpha_{g_t}$ for which $|t|< \delta_n$, will then yield a well defined operator from $\qM(\frac{n}{n+1}K)$ into $\qM(\frac{n}{n+1}K+g(t))\subset \qM(K)$. So the subalgebra $\cup_{n=1}^\infty\qM(\frac{n}{n+1}K)$ of $\qM(K)$ then represents a space of observables inside $\qM(K)$, for which the restriction of the operators $\alpha_{g(t)}$ to the context of $\qM(K)$ yield a well defined dynamics for short times along this contour (where the ``shortness'' of the time depends on the specific observable). If now the algebra $\qM$ was 
additive in the sense of \cite[Definition 4.13]{araki}, we would have that $[\cup_{n=1}^\infty\qM(\frac{n}{n+1}K)]''=\qM(K)$. That is for additive systems, the subalgebra of $\qM(K)$ for which we obtain dynamics for short times, is weak* dense in $\qM(K)$. 

We summarise the conclusions of the above discussion in the following theorem:

\begin{theorem}Let $\qM(M)$ be a local algebra tangentially conditioned to $\qM(\mathbb{M})$. If $\qM(\mathbb{M})$ is additive in the sense of \cite[Definition 4.13]{araki}, then for any smooth contour $C$ through $p\in M$, the dynamics along this contour will for a small enough neighbourhood $\cO_p\subset M$ of $p$, lift to dynamics for short times on a weak* dense subalgebra of $\qM(\cO_p)$. The dynamics at the algebra level is determined by the local action of the translation automorphisms on $\qM(\mathbb{M})$.
\end{theorem}

(For these same ideas to work in the $C^*$-algebra context, we need our local algebra to satisfy what we might call ``strong'' additivity in the sense that for any open double cone $K$ and any collection of open subsets $\{\cO_\lambda\}$ of $K$ covering $K$, we will need $\cup_{n=1}^\infty\cA(\cO_\lambda)$ to be norm-dense in $\cA(K)$.)

\subsection{The space of generators of local flows}

Each locus of points of either the form $(t, 0, 0, \dots)$ or the form $(0, 0, \dots, 0, x_k, 0, \dots, 0)$ is a copy of $\mathbb{R}$. We may denote these loci by $\mathbb{R}_t$ and $\mathbb{R}_k$ ($2\leq k\leq d$) respectively. The groups $\alpha_x$ corresponding to translation by $x$, where $x$ belongs to either $\mathbb{R}_t$ or $\mathbb{R}_k$, are then one-parameter groups on $\qM(\mathbb{M})$.

Hence the derivatives at 0, namely $\delta_t$ and $\delta_k$, are densely defined closed *-derivations. For the sake of simplicity, we will in the discussion hereafter write $\delta_0$ for $\delta_t$. Our first result in this subsection, shows that these derivations are the appropriate noncommutative analogues of the partial differential operators $\frac{\partial}{\partial x_i}$. This result also shows that in a very real sense, the space  $\mathrm{span}\{\delta_k: 0\leq k\leq (d-1)\}$ acts as a space of infinitesimal generators of the action of the translation automorphisms on $\qM(\mathbb{M})$.

\begin{theorem}\label{dspace}
Let $C$ be a smooth contour through $0\in \mathbb{M}$ parametrised by say $x(t)$ where $-1\leq t\leq 1$, and $x(0)=0$. Let $\cO$ be a neighbourhood of $0\in \mathbb{M}$. Then for any $f\in\qM(\cO)\cap[\cap_{k=0}^{(d-1)}\mathrm{dom}(\delta_k)]$, the derivative at 0 of the set $\alpha_{x(t)}(f)$ ($t\in [-1,1]$) exists in the weak* topology on $\qM(\mathbb{M})$, and corresponds to $\sum_{k=0}a_k\delta_k(f)$ where $x'(0)=(a_0, a_1, \dots, a_{(d-1)})$.
\end{theorem} 

\begin{proof}
Let $f$ be as in the hypothesis. Since each $\delta_k$ is adjoint-preserving, we may clearly assume that $f=f^*$. As in the previous proof we will write $t\langle k \rangle$ for the vector with $t\in \mathbb{R}$ in the $k$-th coordinate and 0's elsewhere. The fact that the translation automorphisms are implemented by a strongly continuous unitary group acting on the underlying Hilbert space, ensures that for each $k$, we have that $\frac{1}{t}[\alpha_{t\langle k\rangle}(f)-f]$ converges strongly to $\delta_k(f)$ as $t\to 0$.
 
Let $x(t)$ be of the form $x(t)=(x_0(t), x_1(t), \dots, x_{(d-1)}(t))$. We denote the vector $(0,\dots, 0, x_k(t), 0, \dots, 0)$ by $\hat{x}_k(t)$. We first prove that for any $\xi$ in the underlying Hilbert space and any $0\leq k\leq (d-1)$, we have that 
$\frac{1}{t}(\alpha_{\hat{x}_k(t)}(f)-f)\xi$ converges in norm to $x'_k(0)\delta_k(f)\xi=a_k\delta_k(f)\xi$. 

If $x'_k(0)=a_k\neq 0$, there must exist some $\varepsilon>0$ such that $\frac{x_k(t)}{t}\neq 0$ for every $0<|t|<\varepsilon$. Since $x_k(t)\to 0$ as $t\to 0$, the claim will in this case follow by rewriting $\frac{1}{t}(\alpha_{\hat{x}_k(t)}(f)-f)\xi$ as 
$\frac{x_k(t)}{t}[\frac{1}{x_k(t)}(\alpha_{\hat{x}(t)}(f)-f)\xi]$ for all $0<|t|<\varepsilon$, and then letting $t\to 0$.

If $a_k=0$, then for the claim to be true, we must have that $\frac{1}{t}(\alpha_{\hat{x}_k(t)}(f)-f)\xi \to 0$ as $t \to 0$. Suppose this is not the case. In that case there must exist some $\varepsilon > 0$ and a sequence $\{t_n\}$ tending 0 such that $\|\frac{1}{t_n}(\alpha_{\hat{x}_k(t_n)}(f)-f)\xi\|\geq \varepsilon$ for all $n$. Since $(\alpha_{0}(f)-f)\xi=0$, we must then also have that $x_k(t_n)\neq 0$ for all $n$. But then we may write 
$\frac{1}{t_n}(\alpha_{\hat{x}_k(t_n)}(f)-f)\xi$ as $\frac{x_k(t_n)}{t_n}[\frac{1}{x_k(t_n)}(\alpha_{\hat{x}(t_n)}(f)-f)\xi]$. But as $n\to \infty$, this expression must converge in norm to $a_k\delta_k(f)\xi=0$, which is a clear contradiction. Hence our assumption that $\frac{1}{t}(\alpha_{\hat{x}_k(t)}(f)-f)\xi$ does not converge to 0 as $t\to 0$, must be false.

We claim that for each $k$, the terms $\Pi_{i=k+1}^{(d-1)}\alpha_{\hat{x}_i(t)}([\frac{1}{t}(\alpha_{\hat{x}_k(t)}(f)-f)]-a_k\delta_k(f))$ converge to 0 in the weak* topology as $t\to 0$. 

Suppose that for some $k$ this is not the case. Then there must exist a normal state $\omega$ of $\qM(\mathbb{M})$ such that $\omega(\Pi_{i=k+1}^{(d-1)}\alpha_{\hat{x}_i(t)}([\frac{1}{t}(\alpha_{\hat{x}_k(t)}(f)-f)]-a_k\delta_k(f)))\not\to 0$, or equivalently there exists a normal state $\omega$ and a sequence $\{t_n\}$ tending to 0 such that for some $\widetilde{\varepsilon}>0$ we have that
$$|\omega(\Pi_{i=k+1}^{(d-1)}\alpha_{\hat{x}_i(t_n)}([\frac{1}{t_n}(\alpha_{\hat{x}_k(t_n)}(f)-f)]-a_k\delta_k(f)))|\geq \widetilde{\varepsilon}$$for 
each $n\in \mathbb{N}$. Note that with $w(k,t)$ denoting the vector $(0,\dots, 0, x_{k+1}(t),\dots, x_{(d-1)}(t)$, we have that  $\Pi_{i=k+1}^{(d-1)}\alpha_{\hat{x}_i(t)}=\alpha_{w(k,t)}$ for each $k$. For the sake of simplicity we will in the ensuing argument make 
these substitutions. Since for each $\xi$ in the underlying Hilbert space we have that 
$[\frac{1}{t_n}(\alpha_{\hat{x}_k(t_n)}(f)-f)-a_k\delta_k(f)]\xi\to 0$ as $n\to \infty$, we know from the Banach-Steinhaus theorem that 
$\sup_n\|\frac{1}{t_n}(\alpha_{\hat{x}_k(t_n)}(f)-f)-a_k\delta_k(f)\|<\infty$. But if that is the case then by 
\cite[Proposition 2.4.1]{BR}, $[\frac{1}{t_n}(\alpha_{\hat{x}_k(t_n)}(f)-f)-a_k\delta_k(f)]$ must in fact converge to 0 in the 
$\sigma$-strong topology, and not just strongly. 
Since by assumption $f=f^*$, the convergence is actually in the $\sigma$-strong* topology. Now recall that $\alpha_{w(k,t_n)}$ converges strongly to the identity map on $\qM(\cO)$ as $n\to \infty$. By \cite[III.3.2.2]{Bla}, this convergence also takes place in the weak* 
($\sigma$-weak) topology. So for any normal state $\nu$ on $\qM(\mathbb{M})$, each $\nu\circ \alpha_{w(k,t_n)}$ is again a normal state with $\lim_{n\to\infty}\nu(\alpha_{w(k,t_n)}(a))=\nu(a)$ for every $a\in \qM(\mathbb{M})$. If now we apply \cite[Proposition III.5.5]{Tak} we will have that
$$\lim_{n\to\infty}\omega(\alpha_{w(k,t_m)}([\frac{1}{t_n}(\alpha_{\hat{x}_k(t_n)}(f)-f)]-a_k\delta_k(f)))=0$$uniformly in $m\in\mathbb{N}$. But this contradicts our earlier assumption that 
$$|\omega(\alpha_{w(k,t_n)}([\frac{1}{t_n}(\alpha_{\hat{x}_k(t_n)}(f)-f)]-a_k\delta_k(f)))|\geq\widetilde{\varepsilon}$$for each $n\in\mathbb{N}$. Hence the assumption that for some $k$ the net 
\newline $\alpha_{w(k,t)}([\frac{1}{t}(\alpha_{\hat{x}_k(t)}(f)-f)]-a_k\delta_k(f))$ is not weak* convergent to 0, must be false. Therefore each $\alpha_{w(k,t)}([\frac{1}{t}(\alpha_{\hat{x}_k(t)}(f)-f)]-a_k\delta_k(f))$ is weak* convergent to 0 as $t\to 0$. Moreover by \cite[III.3.2.2]{Bla}, we also have that each $a_k\alpha_{w(k,t)}(\delta_k(f))$ is weak* convergent to $a_k\delta_k(f)$ as $t\to 0$. 

We are now ready to prove the primary claim of the theorem. To do this we simply note that 
$$\frac{1}{t}(\alpha_{x(t)}(f)-f)-\sum_{k=0}^{(d-1)}a_k\delta_k(f)$$may be rewritten as $$\sum_{k=0}^{(d-1)} 
\Pi_{i=k+1}^{(d-1)}\alpha_{\hat{x}_i(t)}[\frac{1}{t}(\alpha_{\hat{x}_k(t)}(f)-f)-a_k\delta_k(f)] 
+\sum_{k=0}^{(d-1)}a_k[\Pi_{i=k+1}^{(d-1)}\alpha_{\hat{x}_i(t)}(\delta_k(f))-\delta_k(f)]$$and then apply the foregoing conclusions to see that $\frac{1}{t}(\alpha_{x(t)}(f)-f)$ converges to $\sum_{k=0}a_k\delta_k(f)$ in the weak* topology as $t\to 0$.
\end{proof}

Having identified the objects that serve as the quantum analogues of the partial differential operators $\frac{\partial}{\partial x_i}$, it is natural to then use these objects to identify a subalgebra of $\qM(\mathbb{M})$, which is the quantum analogue of $C^\infty(M)$. This is done in the next theorem, which also shows that this subalgebra of ``smooth'' elements of $\qM(\mathbb{M})$, is in fact a weak* dense subalgebra.

\begin{theorem}\label{thm4.3}
For any local algebra $\qM(\mathbb{M})$ satisfying the strong operator continuity assumption regarding the translation automorphisms, the algebra $\qM^\infty(\mathbb{M})=\{a\in \qM(\mathbb{M}): a\in \mathrm{dom}(\delta_{\pi(1)}\dots\delta_{\pi(k)}),\ k\in \mathbb{N},\ 0\leq \pi(i)\leq (d-1)\}$ is weak* dense in $\qM(\mathbb{M})$.
\end{theorem}

We pause to point out that $\qM^\infty(\mathbb{M})$ can in a very natural way be described as a Fr\'echet space (see the discussion preceding \cite[Theorem 2.2.3]{Br}). It is also an exercise to see that each element of the space $\mathrm{span}\{\delta_k: 0\leq k\leq (d-1)\}$ will map $\qM^\infty(\mathbb{M})$ back into itself.

\begin{proof}
The proof is based on a modification of \cite[Theorem 2.7]{pazy}. Hence at some points we will not give full details, but instead refer the reader to the corresponding argument in \cite{pazy}. 

Let $\mathfrak{C}$ be the space of all complex-valued $C^\infty$ functions on the open cell $(0,\infty)^d$ which are compactly supported. Given any $x\in \mathbb{M}$, we will throughout write $\alpha_x$ for the automorphism corresponding to translation by $x$ in the sense that $\alpha_x(\qM(\cO))=\qM(x+\cO)$. We will further write $t\langle k \rangle$ for the vector with $t\in \mathbb{R}$ in the $k$-th coordinate and 0's elsewhere, and $\varphi_k$ for the partial derivative of $\varphi$ with respect to the $k$-th coordinate.
Given any $a\in \qM(\mathbb{M})$ and $\varphi\in \mathfrak{C}$, then with $\mathbf{s}$ denoting $(s_0,\dots,s_{(d-1)})$, we set
$$a(\varphi)=\int_0^\infty\dots\int_0^\infty\varphi(\mathbf{s})\alpha_{\mathbf{s}}(a)\,ds_0\,ds_1\dots ds_{(d-1)}.$$In view of the assumption regarding the strong operator continuity of the representation $x\to \alpha_x$, this integral converges in the strong operator topology, which in turn ensures that $a(\varphi)\in \qM(\mathbb{M})$. Assuming that $h>0$ it is then an exercise to see that for example
\begin{eqnarray*}
&&\frac{1}{h}(\alpha_{h\langle 0\rangle}(a(\varphi))-a(\varphi))\\ &=& \int_0^\infty\dots\int_0^\infty\tfrac{1}{h}[\varphi(s_0-h,s_1\dots,s_{(d-1)})-\varphi(s_0,s_1\dots,s_{(d-1)})]\alpha_{\mathbf{s}}(a)\,ds_0\,ds_1\dots ds_{(d-1)}\\
&\to& (-1)\int_0^\infty\dots\int_0^\infty\varphi_{s_0}(\mathbf{s})\alpha_{\mathbf{s}}(a)\,ds_0\,ds_1\dots ds_{(d-1)}.
\end{eqnarray*}
(Again convergence is in the strong operator topology.) Hence $a(\varphi)\in \mathrm{dom}(\delta_0)$ with $\delta_0(a(\varphi))= (-1)a(\varphi_{s_0})$. Here we chose the coordinate $k=0$ for no other reason than simplicity of notation. Hence for any $k$ we have that $a(\varphi)\in \mathrm{dom}(\delta_k)$ with $\delta_k(a(\varphi))=(-1)a(\varphi_k)$. This in particular ensures that the space $Y=\mathrm{span}\{a(\varphi): a\in \qM(\mathbb{M}), \varphi\in \mathfrak{C}\}$ is contained in $\qM^\infty(\mathbb{M})$. To conclude the proof, we therefore need to show that $Y$ is weak* dense in $\qM(\mathbb{M})$. 

Suppose that this is not the case. Then by the Hahn-Banach theorem there must exists some non-zero element $\rho$ of 
$(\qM(\mathbb{M}))_*$, which vanishes on $Y$. But each such $\rho$ is strong operator continuous \cite[Theorem III.2.1.4]{Bla}. Hence for each $a\in \qM(\mathbb{M})$ and each $\varphi\in\mathfrak{C}$ we have that 
$$0=\rho(a(\varphi))= \int_0^\infty\dots\int_0^\infty\varphi({\mathbf{s}})\rho(\alpha_{\mathbf{s}}(a))\,ds_0\,ds_1\dots ds_{(d-1)}.$$This in turn ensures that for each $a\in \qM(\mathbb{M})$, the strong operator continuous map $\mathbf{s}\to \rho(\alpha_{\mathbf{s}}(a))$ is identically 0 on the open cell $(0,\infty)^d$. Letting $\mathbf{s}\to (0, \dots,0)$, it follows that $\rho(a)=0$ for all $a\in \qM(\mathbb{M})$. But this contradicts the assumption that $\rho\neq 0$. We must therefore have that $Y$ is weak*-dense in $\qM(\mathbb{M})$. This concludes the proof.
\end{proof} 

Having thus identified the appropriate noncommutative analogues of the partial differential operators $\frac{\partial}{\partial x_i}$ and of $C^\infty(\mathbb{M})$, we are now in a position to identify the objects that may be regarded as quantum ``local flows'' on 
$\mathbb{M}$. We have already noted that for a general manifold $M$, $C^\infty$ local flows are in a 1-1 correspondence with the derivations on $C^\infty(M)$. (Combine \cite[Remark 2.3.11(1)]{Thi} with the comment preceding \cite[Theorem 2.2.24]{Thi}.) Taking 
the insightful work of Bratteli \cite{Br} as a point of reference for further development, one may on this basis propose the space of all weak*-closable weak*-densely defined derivations on $\qM(\mathbb{M})$ that map $\qM^\infty(\mathbb{M})$ back into itself, as the 
quantum analogue of local flows on $\mathbb{M}$. (See for example the introduction to \cite[Chapter 2]{Br}.) We shall denote this space of derivations by $\Delta_{\mathbb{M}}$. Since the space $\qM^\infty(\mathbb{M})$ was constructed using a finite-dimensional space of 
derivations, results like \cite[Theorem 2.3.6]{Br} (due to Batty) suggest that the elements of $\Delta_{\mathbb{M}}$ are not far from being generators of groups of transformations on $\qM(\mathbb{M})$, which ties in well with the classical theory. The local 
correspondence of tangentially conditioned algebras to $\qM(\mathbb{M})$, allows one to then at least on chart-wise basis attempt to lift these ideas from $\qM(\mathbb{M})$ to $\qM(M)$.

\subsection{Quantum graded algebras for $\qM(\mathbb{M})$}
We shall here discuss the potential utility of the pair $(\qM^\infty(\mathbb{M}), \Delta_{\mathbb{M}})$ for the construction of graded algebras, where $\qM^\infty(\mathbb{M})$ is as defined in Theorem \ref{thm4.3}, and $\Delta_{\mathbb{M}}$ as defined in the final paragraph of the previous section. Note that one may also view $\qM^\infty(\mathbb{M})$ as the subalgebra of elements of $\qM(\mathbb{M})$ which are ``smooth'' with respect to the action of the translation group. The discussion at the end of the previous subsection, shows that the space of derivations $\Delta_{\mathbb{M}}$ in some sense provides the technology for giving a chart-wise description of the ``quantum local flows'' of a given tangentially conditioned algebra $\qM(M)$. We have in addition seen that under mild restrictions, the space $\qM^\infty(\mathbb{M})$ is weak* dense in $\qM(\mathbb{M})$. Hence a quantum graded algebra of differential forms for $\mathbb{M}$ constructed using these objects, should in principle be ``chart-wise'' relevant for tangentially conditioned algebras $\qM(M)$. 

Using the work of Michel du Bois-Violette as a template (see the excellent review in \cite{dje}), one may now construct a graded algebra of differential forms from the pair $(\qM^\infty(\mathbb{M}), \Delta_{\mathbb{M}})$. The actual construction of such a quantum graded algebra, can be done exactly as in section 2.5 of \cite{dje}, with the only difference being that we replace the pair $(\qM, \mathrm{Der}_{\qM})$ used by Djemai, et al, by the pair $(\qM^\infty(\mathbb{M}), \Delta_{\mathbb{M}})$. All other ingredients remain exactly the same. We pause to justify this replacement before going on to explain why this framework is sufficient for the construction to go through.

Recall that by assumption the translation automorphisms are induced by a strongly continuous unitary group acting on the underlying Hilbert space. Using this fact, one is able to conclude that each of the derivations $\delta_k$ is of the form $\delta_k(a)= i[H_k,a]$. For the time variable $H_0$ is just the Hamiltonian, with $H_k$ being a momentum operator for $k= 1,\dots, (d-1)$. These operators are all necessarily unbounded, and hence so are each of the $\delta_k$'s. These are therefore clearly not defined on all of 
$\qM(\mathbb{M})$, and hence if we want a model incorporating the information encoded in the $\delta_k$'s, we cannot a priori insist on 
using $\qM(\mathbb{M})$ in the construction, as this will exclude these operators. However each $\delta_k$ is defined everywhere on a smooth part of the algebra, namely $\qM^\infty(\mathbb{M})$. The replacement of $\qM(\mathbb{M})$ with $\qM^\infty(\mathbb{M})$, therefore allows one to incorporate the $\delta_k$'s into the picture. In fact Djemai himself reveals an implicit concern for ``smoothness'', when at the start of \cite[Subsection 2.5.2]{dje} he points out that the constructs described at that point may be applied to $C^\infty(M)$.

For the readers who are concerned about the validity of the claim that the construction in \cite[\S 2.5]{dje} carries over to the pair $(\qM^\infty(\mathbb{M}), \Delta_{\mathbb{M}})$, we hasten to point out that the construction in \cite[\S 2.5]{dje} is entirely algebraic, and that all we in principle need is a $*$-algebra, and a space of derivations on that algebra which admits a left-module action of that algebra, and that we do have. We pause to further justify the claim that a $*$-algebra rather than a $C^*$-algebra will suffice. Note for example that although the author invokes the ``topological'' tensor product at the start of section 2.2 of \cite{dje}, the algebraic tensor product will do just as well for this part of the construction. Note further that in the construction described in sections 2.1-2.5 of \cite{dje}, there are three crucial ingredients. These are Property 1 on page 808, Property 2 on page 808, and Proposition 3 on page 809. Although full details are not given in the actual text of \cite{dje}, it can be seen from \cite[Chapter III, \S X]{bour} that Property 1 is a purely algebraic property. For the other two aspects Djemai cites \cite{Con} as a reference. In that paper Connes announces a 7 step programme (on p 264), with the bulk of \cite{Con} devoted to step II. This is the part required by the construction in \cite{dje}. But as can be seen from the declaration at the top of page 262 of \cite{Con}, the content of \cite[Part II]{Con} is purely algebraic! Readers that have some concern that at some point Part II of \cite{Con} has a hidden reliance on the more topological Part I, will be reassured by the discussion on page 310, where Alain Connes describes the relationship between Parts I and II. 

The above discussion leads us to the following conclusion:

\begin{theorem} The algebra $\qM(\mathbb{M})$ admits the construction of a smooth quantum graded algebra of differential forms which on a chart-wise basis models the action of smooth quantum local flows on a tangentially conditioned algebra $\qM(M)$.
\end{theorem}

In closing we wish to point out that there is a very comprehensive theory of the derivational approach to noncommutative differential geometry. Our goal in this paper was to demonstrate how one aspect of this theory may be incorporated into the theory of local algebras. There is surely more that can be done in this regard, but that is not the concern of the present paper.

\section{Conclusions.}
In the present paper we have continued our study on the new approach based on Orlicz spaces to the analysis of large systems, i.e. systems having an infinite number of degrees of freedom, see \cite{LM}, \cite{LM2}, \cite{ML}. Having shifted our focus to quantum field theory, we here show that this new strategy initially developed for quantum statistical mechanics, can potentially in a very natural and elegant way be applied to quantum field theory, obviously with suitable modifications. It is important to note that the modifications necessary for the application of this strategy to quantum field theory, are drawn from the very basic ingredients of the principle of relativity; see Corollary \ref{1}.

Having thus established the static setting of quantum fields, a full analysis of physical laws and the mathematical equations describing them, also demands differential structures. By way of example one may note that even Maxwell equations fit naturally into differential geometry; specifically the calculus on manifolds. In addressing this issue, we have shown that using a modified du Bois-Violette approach to non-commutative differential geometry, the action of the Poincar\'e group on tangentially conditioned local algebras, allows for the construction of graded algebras of differential forms for these algebras. However, as was hinted at at the end of the previous section, the theory of quantum local flows on local algebras is incomplete.  


\newpage
\begin{thebibliography}{99}

\bibitem{AIT} J-P Antoine, A. Inoue, C. Trapani \textit{Partial $^*$-Algebras and Their Operator Realizations} Mathematics and Its Applications, vol. \textbf{553}, Kluwer, Dordrecht, NL, 2002

\bibitem{araki} H. Araki, \textit{Mathematical Theory of Quantum Fields}, Oxford University Press

\bibitem{Bag} F. Bagarello, Algebras of unbounded operators and physical applications: A survey. \textit{Rev. Math. Phys.} \textbf{19} 231 (2007).

\bibitem{BaF} C B\"ar \& K Fredenhagen (Editors), \textit{Quantum Field Theory for curved spacetimes :Concepts and mathematical foundations}, Lecture Notes in Physics Vol 178, Springer-Verlag, Berlin-Heidelberg, 2009.

\bibitem{BS} G Bennet and R Sharpley, \textit{Interpolation of Operators}, Academic Press, London, 1988.

\bibitem{BW1} J. J. Bisognano, E. H. Wichmann, On the duality condition for Hermitean scalar fields, \textit{J. Math. Phys.} \textbf{16}, 985 (1975)

\bibitem{BW2} J. J. Bisognano, E. H. Wichmann, On the duality condition for quantum fields, \textit{J. Math. Phys.} \textbf{17}, 303 (1976).

\bibitem{Bla} B Blackadar, \textit{Operator Algebras}, Springer, 2006.

\bibitem{Bor1} H. J. Borchers, Algebraic aspects of Wightman quantum field theory in \textit{International Symposium on Mathematical Methods in Theoretical Physics}, ed. H. Araki, Lecture Notes in Mathematics, vol. \textbf{39}, Springer, (1975) pp 283-292.

\bibitem{Bor3} H. J. Borchers, Energy and momentum as observables in quantum field theory, \textit{Commun. Math. Phys.} \textbf{2} 49-54 
(1966)

\bibitem{Bor2} H. J. Borchers, On revolutionizing quantum field theory with Tomita's modular theory, \textit{J. Math. Phys.} \textbf{41}, no. 6, 3604-3673 (2000).

\bibitem{BY} H. J. Borchers, J. Yngvason, Positivity of Wightman functionals and the existence of local nets, \textit{Commun. Math. Phys.} \textbf{127} 607-615 (1990)

\bibitem{bour} N Bourbaki, \textit{Alg\`ebre I}, Hermann, Paris, 1970

\bibitem{Br} O Bratteli, \textit{Derivations, Dissipations and Group Actions on $C^\ast$-algebras}, Springer-Verlag, Berlin-Heidelberg-New York, 1986.

\bibitem{BR} O. Bratteli, D. Robinson, \textit{Operator algebras and Quantum Statistical Mechanics }, Texts and Monographs in Physics, Springer Verlag; vol. I, 1979; vol. II, 1981

\bibitem{BF} R Brunetti, K Fredenhagen, Algebraic Approach to Quantum Field Theories. in: \textit{Encyclopedia of Mathematical Physics}, Amsterdam: Elsevier, 2007. - ISBN: 9780125126601 

\bibitem{BFV}  R Brunetti, K Fredenhagen, R Verch, The generally covariant locality principle-a new paradigm for local quantum field theory, \textit{Comm. Math. Phys.} \textbf{237} (2003) no. 1-2, 31--68.

\bibitem{buch1} D. Buchholz, On quantum fields that generate local algebras, \textit{J. Math. Phys.} \text{31} 1839-1846 (1990)

\bibitem{BH} D. Buchholz, R. Haag, The quest for understanding in relativistic quantum physics, \textit{J. Math. Phys.} \textbf{41}, 3674-3697 (2000).

\bibitem{Con} A Connes, Noncommutative differential geometry, \textit{Inst. Hautes Etudes Sci. Publ. Math.}, No. 62 (1985), 257--360.

\bibitem{D1} J. Dimock, Algebras of local observables on a manifold, \textit{Comm. Math. Phys.} \textbf{77} (1980), 219--228.

\bibitem{D2} J. Dimock, Dirac quantum fields on a manifold, \textit{Trans. Amer. Math. Soc.} \textbf{269} (1982) no. 1, 133--147.

\bibitem{dje} A E F Djemai, Introduction to Dubois-Violette's noncommutative differential geometry, \textit{International Journal of Theoretical Physics} \textbf{34} (1995) No. 6, 801--887.

\bibitem{DSW} W. Driessler, S. J. Summers, E. H. Wichmann, On the connection between quantum fields and von Neumann algebras of local operators, \textit{Commun. Math. Phys.} \textbf{105}, 49-84 (1986)

\bibitem{DF} W. Driessler and J. Fr\"ohlich, The reconstruction of local observable from the euclidean Green's functions of relativistic quantum field theory, \textit{Annales de l'I.H.P., section A} \textbf{27}(3)(1977), 221-236.

\bibitem{FK} T. Fack, H. Kosaki, Generalized $s$-numbers of $\tau$-measurable operators, \textit{Pac. J. Math.} \textbf{123} (1986), 269-300

\bibitem{GQ} J. Gallier and J. Quaintance, \textit{Notes on Differential Geometry and Lie Groups, I \& II}, in progress.(See \verb+ http://www.cis.upenn.edu/~jean/gbooks/manif.html +)

\bibitem{GW1} L. G\"arding, A. S. Wightman, Representation of the anticommutation relations, \textit{Proc. Nat. Acad. Sci.} \textbf{40}, 617 (1954)

\bibitem{GW2} L. G\"arding, A. S. Wightman, Representation of the commutation relations, \textit{Proc. Nat. Acad. Sci.} \textbf{40}, 622 (1954)

\bibitem{Gol} S Goldstein, Conditional expectation and stochastic integrals in non-commutative $L^p$ spaces', \textit{Math Proc Camb Phil Soc} \textbf{110}(1991), 365--383

\bibitem{haag} R. Haag, \textit{Local Quantum Physics. Fields, Particles, Algebras}, Springer, 2nd edition, 1996

\bibitem{uffe1} U. Haagerup, Normal weights on W*-Algebras, \textit{J. Funct. Anal.} \textbf{19} 302-317 (1975)

\bibitem{uffe-DW1} U. Haagerup, On the dual weights for crossed products of von Neumann algebras I, \textit{Math. Scandinavica} 
\textbf{43} (1979), pp. 99 - 118

\bibitem{uffe2} U. Haagerup, Operator valued weights in von Neumann algebras: I, \textit{J. Funct. Anal.} \textbf{32} 175-206 (1979)

\bibitem{uffe3} U. Haagerup, Operator valued weights in von Neumann algebras: II, \textit{J. Funct. Anal.} \textbf{33} 339-361 (1979)

\bibitem{HJX} U. Haagerup, M. Junge, Q. Xu, A reduction method for noncommutative $L_p$-spaces and applications, \textit{TAMS}, \textbf{362} 2125-2165 (2010)

\bibitem{hor} S.S. Horuzhy, \textit{Introduction to Algebraic Quantum Field Theory}, Springer, 1990.

\bibitem[KRi] {KRi} RV Kadison and JR Ringrose, \textit{Fundamentals of the Theory of Operator Algebras: Vol 2}, Academic Press, New York, 1986.

\bibitem{kos} H. Kosaki, Applications of complex interpolation method to a von Neumann algebra (Non-commutative $L^p$-spaces) \textit{J. Funct. Anal.} \textbf{56} 29-78 (1984)

\bibitem{L} LE Labuschagne, A crossed product approach to Orlicz spaces, \textit{Proc LMS} \textbf{107} (3) (2013), 965-1003.

\bibitem{LM} LE Labuschagne, WA Majewski, Maps on non-commutative Orlicz spaces, \textit{Illinois J. Math}. {\bf 55}, 1053-1081, (2011)

\bibitem{LM2} LE Labuschagne and WA Majewski, Quantum dynamics on Orlicz spaces, arXiv:1605.01210 [math-ph].

\bibitem{ML} WA Majewski, LE Labuschagne, On applications of Orlicz spaces to Statistical Physics, \textit{ Ann. H. Poincare.}, {\bf 15}, 1197-1221, (2014)

\bibitem{ML2} W. A. Majewski, L. E. Labuschagne, On Entropy for general quantum systems (arXiv:1804.05579 [math-ph]).

\bibitem{Maj1} W. A. Majewski, L. E. Labuschagne, Why are Orlicz spaces useful for Statistical Physics? in \textit{Noncommutative Analysis, Operator Theory and Applications;} Eds. D. Alpay et al. Birkhauser-Basel, Series: Linear Operators and Linear Systems, vol 252, 271-283 (2016)

\bibitem{studyguide} W. A. Majewski, On quantum statistical mechanics; A study guide,  \textit{Adv. Math. Phys} , Article ID 9343717 (2017);  arXiv 1608.06766v2 [math-ph]

\bibitem{nelson} E. Nelson, Notes on non-commutative integration, \textit{J. Funct. Anal.} \textbf{15} (1974), 103.

\bibitem{Nel} E. Nelson, Construction of Quantrum Fields from Markov  Fields, \textit{J. Funct. Anal.} \textbf{12}(1973), 211.

\bibitem{pazy} A Pazy, Semigroups of linear operators and applications to partial differential equations, Springer-Verlag, New York, 1983

\bibitem{PT} G. Pedersen and M. Takesaki, The Radon-Nikodym theorem for von Neumann algebras, \textit{Acta. Math.}\textbf{130}, 53-87, (1973)

\bibitem{RSI} M. Reed, B. Simon, \textit{Methods of Modern Mathematical Physics. I. Functional Analysis} Academic Press, New York and London, 1972.

\bibitem{RS2}  M Reed \& B Simon, \textit{Fourier Analysis, Self-Adjointness : Methods of Modern Mathematical Physics, Vol. 2 (1st Ed)}, Academic Press, 1975.

\bibitem{ReSch} H. Reeh, S. Schlieder, Bemerkungen zur Unit\"ar\"aquivalenz von Lorentzinvarianten Feldern, \textit{Nuovo Cim.} \textbf{19} 787-793 (1961)

\bibitem{Sak} S. Sakai, \textit{Operator Algebras in Dynamical Systems}, Cambridge University Press, Cambridge, 1991.

\bibitem{Schmud} K. Schm{\"u}dgen, \textit{Unbounded Operator Algebras and Representation Theory} Akademie-Verlag, Berlin, 1990

\bibitem{Segal} I. E. Segal, Postulates for General Quantum Mechanics, \textit{Ann. Math.} \textbf{48} 930-948 (1947).

\bibitem{Simon} B. Simon, \textit{The $P(\phi)_2$ Euclidean (Quantum) Field Theory}, Princeton University Press, Princeton, New Jersey, 1974

\bibitem{Tak} M Takesaki, \textit{Theory of Operator Algebras, Vol I, II, III}, Springer, New York, 2003.

\bibitem{terp} M. Terp, {\it $L^p$ spaces associated with von Neumann algebras}. K{\o}benhavs Universitet, Mathematisk Institut, Rapport No 3a (1981).

\bibitem{Thi} W. Thirring, {\it A Course in Mathematical Physics I: Classical Dynamical Systems (1st edition)}, Springer-Verlag, 1978.

\bibitem{Wal} R. M. Wald, {\it General Relativity}, University of Chicago Press, 1984.

\bibitem{weh} A. Wehrl, General properties of entropy, {\em Rev. Mod. Phys.}, {\bf 50}, 221-260, 1978
 
\bibitem{W} A. S. Wightman, \textit{Quelque probl\`emes math\'ematique de la th\'eorie quantique relativiste}. In: Lecture Notes, 
Facult\'e de Sciences. Univ. de Paris, 1957.

\bibitem{Y} J. Yngvason, The role of type III factors in Quantum Field Theory, \textit{Rep. Math. Phys.} \textbf{55} 135-147 (2005)
\end{thebibliography}
\end{document}